\documentclass[%
 reprint,
nofootinbib,
 amsmath,amssymb,
 aps,
pre,
]{revtex4-2}

\usepackage{amsfonts,amsmath,amsthm,amssymb}

\usepackage{cleveref}
\usepackage{todonotes}

\newtheorem{theorem}{Theorem}
\newtheorem{lemma}{Lemma}

\newcommand{\map}{\textcolor{brown}}

\begin{document}


\title{A Bounded-Confidence Model of Opinion Dynamics with Adaptive Interaction Probabilities}


\author{Leila Thompsky}
\affiliation{University of California, Los Angeles, Los Angeles, California, USA}

\author{Yuexuan (Yolanda) Wu}
\affiliation{University of California, Los Angeles, Los Angeles, California, USA}

\author{Mason A. Porter}
\affiliation{University of California, Los Angeles, Los Angeles, California, USA}
\affiliation{Santa Fe Institute, Santa Fe, New Mexico, USA}

\author{Jiajie Luo}
\affiliation{University of Chicago, Chicago, Illinois, USA}

\date{\today}


\begin{abstract}
Models of opinion dynamics aim to capture how individuals' opinions change when they interact with each other. One well-known model of opinion dynamics is the Deffuant--Weisbuch (DW) model, which is a type of bounded-confidence model (BCM). In the DW model, agents have pairwise interactions, and they are receptive to other agents' opinions when 
their opinions are sufficiently close to each other. In this paper, we extend the DW model by studying it on networks with heterogeneous and adaptive edge weights between pairs of agents. These edge weights govern the 
 interaction probabilities between the agents and 
 thereby encode the idea that people are more likely to communicate with individuals with whom they have previously compromised or had other positive interactions.
  We prove theoretical guarantees of our adaptive edge-weighted DW model's convergence properties, the long-time dynamics of its edge weights, and the model's associated ``effective graph", which is a time-dependent subgraph  
 that includes edges only between agents that are receptive to each other's opinions. We support our theoretical results with numerical simulations of our adaptive edge-weighted 
 DW model on a variety of networks and find that including adaptive edge weights yields different qualitative dynamics for different types of networks.
 In particular, for small confidence bounds, we observe that incorporating adaptive edge weights decreases the convergence time for dense networks but increases the convergence time for sparse networks.
\end{abstract}

\maketitle



\section{Introduction}
\label{introduction}

Interactions allow individuals to share their views and to potentially alter their opinions when they consider the perspectives of others~\cite{social_interactions,wasserman1994}.
Communities of individuals, whether online or offline, can thus act as conduits for social influence and the spread of opinions in a population~\cite{baek2025}.

In the study of opinion dynamics, researchers examine how opinions change with time due to social interactions \cite{noorazar-review,starnini2025,caldarelli2025}. They commonly study whether or not a community achieves a consensus opinion, how opinion clusters (i.e., sets of individuals with similar opinions) form, how many opinion clusters form, and how long opinion clusters take to form or dissolve \cite{noorazar-review}. In agent-based models (ABMs) of opinion dynamics, researchers consider a set of agents that have opinions and can alter them due to interactions with other agents. The interactions between agents often depend on an underlying network structure~\cite{newman2018}, which determines which agents have social and/or communication ties with each other. 
In a network setting, the agents are the nodes of a network and the edges between pairs of nodes indicate that the two associated agents can interact with each other. In the present paper, we consider networks in the form of graphs.

When individuals adjust their opinions, they tend to be more receptive to viewpoints that are similar to their own than to viewpoints that differ significantly from theirs \cite{receptiveness, receptiveness2}. 
Bounded-confidence models (BCMs) encode
this observation by incorporating a ``confidence bound'' that restricts 
how much agents' continuous-valued opinions can differ for them to be receptive to each other \cite{bcm_review, meng2018}. 
Agents with opinions that differ {by at least} a confidence bound are not receptive to each other, so they do not influence each other when they interact.

The two best-known BCMs are the Hegselmann--Krause (HK) model \cite{HK} and the Deffuant--Weisbuch (DW) model \cite{deffuant}. The HK model has synchronous opinion updates, in which all agents potentially update their opinions at each time step. By contrast, the DW model has asynchronous opinion updates, in which one randomly (traditionally, uniformly at random) selects a pair of agents to interact at each discrete time step. The DW model also incorporates a compromise parameter that controls how much agents change their opinions when they update their opinions.

There are many generalizations of the baseline DW and HK models~\cite{starnini2025,bcm_review, meng2018}. For example, researchers have {considered heterogeneous confidence bounds \cite{het_conf}, incorporated noise into opinion updates  \cite{rand_noise}, incorporated signed edges and used negative edges {between agents that are repelled by each other's opinions}  
\cite{repulsion_Martins2010, repulsion_kann2023}, and investigated the effects of polyadic interactions between three or more agents \cite{polyadic_interactions,schawe2022}}. 
The asynchronous {opinion update} rule in the DW model allows researchers to modify how to choose agents to interact in a time step. For instance, to model heterogeneous activity levels of individuals in a community, Li and Porter \cite{grace2022} incorporated heterogeneous probabilities of selecting a particular agent for an interaction. One can also incorporate the fact individuals are more likely to interact with some members of their social circles than others, as social relationships themselves vary in intensity and activity. To model this, one can use heterogeneous edge selection --- with agents with closer opinions more likely to interact with each other --- to study the effects on opinion dynamics of algorithmic biases in social media~\cite{alg_bias}.

Researchers have also examined opinion models with various adaptive features \cite{adaptive_networks}. 
One prevalent type of network adaptivity is network structure --- both topology and edge weights --- that coevolve with opinion dynamics on networks.
Sugishita et al.~\cite{tie_decay_op} studied the DeGroot consensus model on a tie-decay network \cite{tie_decay}, in which ties (which encode social and/or communication relationships between individuals) decay between interactions and increase instantaneously when two agents interact with each other. 
Kozma and Barrat  \cite{kozma2008, kozma2008b}  
{extended} the DW model by incorporating an edge-rewiring process.
In their model, ``discordant" edges (i.e., edges between agents that differ in their opinions by more than the confidence bound) have some chance to break; a broken edge then attaches to a new agent that one selects uniformly at random. 
Kan et al.~\cite{kan_2023} generalized {Kozma and Barrat's} model by considering a homophilic rewiring process, where broken edges are more likely to attach to agents with similar opinions. Additionally, they incorporated a separate tolerance threshold (which differs in general from the confidence bound) to determine how far apart two agents' opinions can be before one treats an edge between them as discordant.
{Brede \cite{brede2019} studied a BCM on an adaptive network where agents rewire to maximize their influence.} 
{Krishnagopal and Porter \cite{krishnagopal2025} studied a BCM 
with a neighborhood-based network adaptivity, in which agents rewire preferentially to other agents whose neighbors' opinions are close to their own.}
Agostinelli et al.~\cite{agost2026} examined a BCM {on an adaptive network} with polyadic interactions to analyze the importance of group interactions in opinion dynamics. 
Researchers have also examined BCMs with adaptive parameters. Li et al.~\cite{bcm_paper} formulated extensions of the DW and HK models with adaptive confidence bounds that increase or decrease due to agents' interactions. A key motivation of their model is that people's receptiveness towards others is influenced by their past interactions. 
Going beyond BCMs, researchers have also studied other opinion models (such as voter models \cite{adaptive_voter,durrett2012,kureh2020,adaptive_voter_hypergraphs,adaptive_voter_noisy} {and the DeGroot model \cite{adaptive_DG}}) on adaptive networks with rewiring processes. 

In the present paper, we generalize the baseline DW model by incorporating time-dependent and heterogeneous edge weights between agents in a network. 
The weight of an edge determines the probability of an interaction between the two agents that are attached to that edge. As in~\cite{bcm_paper}, our adaptive edge weights reflect the idea that a tie between two agents changes as a result of their interactions. However, instead of modeling an agent's receptiveness to another agent, we model how an agent's past interactions affect its interaction patterns and frequencies.
The motivating idea for our model is that people are more likely to interact with someone with whom they have agreed previously than people with whom they have disagreed previously {(see, e.g., \cite{social_exchange})}. 
To encode this idea, we increase the edge weight between two nodes after each interaction in which they update their opinions (this is a ``positive" interaction) and we decrease the edge weight after an interaction that does not lead to opinion updates (i.e., a ``negative" interaction).

We examine the convergence properties of agents’ opinions in our adaptive DW model, and we prove theoretical results about the long-term dynamics of the edge weights. We also analyze the ``effective graphs'' of graphs $G$ that encode relationships between agents.
An effective graph is the time-dependent subgraph of $G$ that includes only the edges in $G$ that are between agents that are receptive to each other. We simulate our adaptive DW model on a variety of graphs and examine the associated major and minor opinion clusters, Shannon entropies, and convergence times. Our numerical simulations yield different behaviors for different types of graphs. 

Our paper proceeds as follows. In Section \ref{sec:baseline_model}, we present the baseline DW model. In Section \ref{sec:model_description}, we describe our DW model with adaptive edge weights. In Section \ref{sec:theoretical_results}, we describe the theoretical results for our model. We then present the results of our numerical simulations in Section \ref{sec:num_simulations} and give concluding remarks in Section \ref{sec:conclusions}. Our code {and additional plots from our numerical simulations} are available at {\url{https://github.com/Yolandayx929/adaptive-edge-BCM}}.


\section{The baseline Deffuant--Weisbuch (DW) model}
\label{sec:baseline_model}

In the baseline Deffuant--Weisbuch (DW) model \cite{deffuant}, pairs of nodes update their opinions asynchronously. 
Consider an undirected,  unweighted, and time-independent network (i.e., graph) $G = (V,E)$. 
Each node $i \in V$ has a time-dependent continuous-valued opinion $x_i(t) \in [0,1]$. At each discrete time $t$, we choose an edge $(i,j) \in E$ uniformly at random. If the opinions of nodes $i$ and $j$ differ by less than the confidence bound $c$ 
(i.e., if $|x_i(t) - x_j(t)| < c$), then nodes $i$ and $j$ have a positive interaction and update their opinions:
\begin{align}
    x_i(t + 1) &= x_i(t) + \mu(x_j(t) - x_i(t)) \nonumber \,,\\
    x_j(t + 1) &= x_j(t) + \mu(x_j(t) - x_i(t)) \,,
\end{align}
where $\mu \in (0,0.5]$ is the compromise parameter. The compromise parameter controls how much the nodes change their opinions due to a positive interaction. However, if the opinions of nodes $i$ and $j$ differ by at least the confidence bound (i.e., if $|x_i(t) - x_j(t)|\geq c$), then the opinions of $i$ and $j$ stay the same: 
\begin{align}
	x_i(t + 1) &= x_i(t)  \nonumber \,, \\
	x_j(t + 1) &= x_j(t)  \,.
\end{align}
The opinions of nodes that do not interact at time $t$ stay the same.

The confidence bound $c$ and the compromise parameter $\mu$ are time-independent and are the same for all nodes. Additionally, at each time $t$, there is an equal (and time-independent) probability of selecting each edge $(i,j) \in E$ for its attached nodes to interact.


\section{DW model with adaptive edge weights}
\label{sec:model_description}

Our adaptive DW model builds on the baseline DW model. Our model, which is an \textit{adaptive edge-weighted DW model}, incorporates heterogeneous interaction probabilities by assigning a time-dependent weight to each edge of an underlying graph $G$. This weight determines the interaction probabilities for each pair of {nodes}. 
The parameters of the opinion dynamics in our adaptive edge-weighted DW model are a confidence bound $c \in [0,1]$, a compromise parameter $\mu \in (0,0.5]$, an initial edge-weight parameter $z_0 \in [0,1]$, an edge-weight increase parameter 
$\gamma\in [0,1]$, and an edge-weight decrease parameter $\delta \in [0,1]$. The graph $G$ can yield additional parameters.

Let $G = (V,E)$ be graph. Each node $i \in V$ has a time-dependent continuous-valued opinion $x_i(t) \in [0,1]$, and each edge $(i,j) \in E$ has a time-dependent weight $z_{ij}(t) \in [0,1]$ that determines the probability that its attached nodes $i$ and $j$ interact. At each discrete time $t$, the probability that nodes $i$ and $j$ interact is 
\[
	P_{ij}(t) = \frac{z_{ij}(t)}{ \sum\limits_{(i',j') \in E} z_{i'j'}(t)} \,.
\]
Nodes update their opinions following the same rule as in the traditional DW model. If the opinions of nodes $i$ and $j$ differ by less than the confidence bound (i.e., $|x_i(t) - x_j(t)| < c$), then they update their opinions: 
  \begin{align}
    x_i(t + 1) &= x_i(t) + \mu(x_j(t) - x_i(t)) \nonumber \,,\\
    x_j(t + 1) &= x_j(t) + \mu(x_j(t) - x_i(t)) \label{eq:opinion_update_positive} \,.
  \end{align} 
However, if the opinions of nodes $i$ and $j$ differ by at least the confidence bound (i.e., $|x_i(t) - x_j(t)| \geq c$), then the opinions of $i$ and $j$ stay the same: 
  \begin{align}
    x_i(t + 1) &= x_i(t)  \nonumber \,, \\
    x_j(t + 1) &= x_j(t)  \label{eq:opinion_update_negative} \,.
  \end{align}
The opinions of nodes that do not interact at time $t$ stay the same. 

We also update the weight of the edge $(i,j)$ between interacting nodes $i$ and $j$.\footnote{Our edge-weight update rule is inspired by the confidence-bound update rule in \cite{bcm_paper}.} If $|x_i(t) - x_j(t)| < c$, then the edge weight $z_{ij}(t)$ increases:
\begin{align} \label{eq:edge_weight_increase}
	z_{ij}(t + 1) = z_{ij}(t) + \gamma(1 - z_{ij}(t)) \,. 
\end{align}
Otherwise, the edge weight $z_{ij}$ decreases:
\begin{align} \label{eq:edge_weight_decrease}
	z_{ij}(t + 1) = \delta z_{ij}(t)\,. 
\end{align}
The weights of edges that are not involved in an interaction at time $t$ stay the same. 


\section{Theoretical results} \label{sec:theoretical_results}

In this section, we state and prove several theoretical properties about our adaptive edge-weighted DW model. 
For a graph $G = (V,E)$, let $N = |V|$ denote the number of nodes and let $m = |E|$ denote the number of edges.
The \emph{limit opinion} $x^i$ of a node $i$ is $\lim\limits_{t\to\infty}x_i(t)$\map{.}
{A \emph{limit opinion cluster} is a maximal set of nodes with the same limit opinion.
That is, nodes $i$ and $j$ are in the same limit opinion cluster when $x^i = x^j$. 
}
Adjacent nodes $i$ and $j$ are \emph{receptive} to each other at time $t$ if the difference between their opinions is less than the confidence bound $c$ (i.e., $|x_i(t) - x_j(t)| < c$).

{In our proofs in this section, we use the notation $\Delta^{ij} := |x^i - x^j|$ for the absolute value of the difference between the 
{limit opinion values} of $i$ and $j$.}

The following theorem by Lorenz \cite{lorenz_stabalization_2005} guarantees that a limit opinion exists for each node of a graph. 

\begin{theorem} [{Lorenz \cite{lorenz_stabalization_2005}}] \label{theorem:stabilization} 
    Let $\{A(t)\}_{t = 0}^\infty \in \mathbb{R}_{\geq 0}^{N \times N}$ be a sequence of row-stochastic matrices. Suppose that each matrix satisfies the following properties:

    \begin{enumerate}
        \item[(1)]{The diagonal entries of $A(t)$ are positive.} 
        \item[(2)]{For each $i, j \in \{1,\dots,N\}$, we have that $[A(t)]_{ij} > 0$ if and only if $[A(t)]_{ji} > 0$.} 
        \item[(3)]{There is a constant $\alpha > 0$ such that the smallest positive entry of $A(t)$ for each $t$ is larger than $\alpha$.} 
    \end{enumerate}

    Given times $t_0$ and $t_1$ with $t_0 < t_1$, let
    \begin{equation}
	    A(t_0, t_1) = A(t_1 - 1) \times A(t_1 - 2) \times \cdots \times A(t_0) \,.
    \end{equation}
    If conditions (1)--(3) are satisfied, then there exists a time $t'$ and pairwise-disjoint classes $\mathcal{I}_1 \cup \dots \cup \mathcal{I}_p = \{1,\dots , N \}$ such that if we reindex the rows and columns of the matrices in the order $\mathcal{I}_1 \cup \dots \cup \mathcal{I}_p$, then
    \begin{equation}
        \lim\limits_{t\to\infty} A(0,t) = \begin{bmatrix} K_1  & & 0 \\ & \ddots & \\ 0 & & K_p \end{bmatrix} \,,
    \end{equation}
    where each $K_q$, with $q \in \{1,2,\ldots, p\}$, is an $|\mathcal{I}_q| \times |\mathcal{I}_q|$ row-stochastic matrix 
    whose rows are all the same. 
\end{theorem}

The opinion updates of the baseline DW model satisfy conditions (1)--(3), so Theorem \ref{theorem:stabilization} establishes that each node's opinion converges to a limit opinion \cite{lorenz_stabalization_2005}. Our adaptive edge-weighted DW model uses the same {opinion update} rule as the baseline DW model, so the opinion of each node in this BCM also converges to a limit opinion.


\subsection{Edge-weight analysis}\label{sec:edge_analysis}

In this subsection, we state and prove theoretical results that describe the behavior of the edge weights in our adaptive edge-weight DW model. Our culminating result is Theorem \ref{theorem:edge_convergence}, which guarantees that{, almost surely (i.e., with probability $1$),}
{the edge weight between two nodes converges to 0 when the nodes are in different limit opinion clusters and converges to 1 when the nodes are in the same limit opinion cluster.}
These results {are analogous to} the confidence-bound results of \cite{bcm_paper}.

\begin{lemma}\label{lemma:infinite_interactions}
    Consider the adaptive edge-weighted DW model with opinion update rule (\ref{eq:opinion_update_positive}, \ref{eq:opinion_update_negative}) and edge-weight update rule (\ref{eq:edge_weight_increase}, \ref{eq:edge_weight_decrease}), and let $(i,j)\in E$. 
    Nodes $i$ and $j$ almost surely interact with each other infinitely many times. 
\end{lemma}

\begin{proof}
    Suppose that nodes $i$ and $j$ interact with each other only a finite number of times. There is then a time $T$ such that nodes $i$ and $j$ do not interact with each other at any time $t\geq T$.  

    Let $t_0 \leq T$ denote the last time that nodes $i$ and $j$ interact, with $t_0 = 0$ if they never interact. Because nodes $i$ and $j$ do not interact at any time $t\geq T$, the weight of the edge between nodes $i$ and $j$ is
    $z_{ij}(t) = z_{ij}(t_0)$ for all $t\geq t_0$. Additionally, for each edge $(i',j')$, it is necessarily true that $z_{i'j'}(t)\leq 1$. 

    Let $P_{ij}$ denote the probability that we select edge $(i,j)$ at time $t$. 
    By the definition of $P_{ij}$, we have 
    \begin{align*}
        P_{ij} &= \frac{z_{ij}(t)}{ \sum\limits_{(i',j') \in E} z_{i'j'}(t)} \\
       		 &= \frac{z_{ij}(t_0)}{ \sum\limits_{(i',j') \in E} z_{i'j'}(t)} \\
       		 &\geq \frac{z_{ij}(t_0)}{{m}}\\ & > 0\,,
    \end{align*}
    where we recall that $m = |E|$.
    
    Therefore, with probability $1$, we select edge $(i,j)$ at some time $t \geq T$. Consequently, with probability $0$, nodes $i$ and $j$ do not interact with each other for any time $t\geq T$. 
    It then follows that nodes $i$ and $j$ almost surely interact with each other infinitely often.
\end{proof}

We now demonstrate the monotonicity\footnote{We say that a sequence $\{a_n\}$ is ``monotone increasing" if $a_n \leq a_{n + 1}$ for all $n$ and that it is ``monotone decreasing" if $a_n \geq a_{n + 1}$ for all $n$.} of the edge weights of our adaptive edge-weighted DW model.

\begin{lemma}\label{lemma:adaptive_DW_influence_T}
    {The} adaptive edge-weighted DW model with opinion update rule (\ref{eq:opinion_update_positive}, \ref{eq:opinion_update_negative}) and edge-weight update rule (\ref{eq:edge_weight_increase}, \ref{eq:edge_weight_decrease})
{has} a time $T$ such that $z_{ij}(t)$ is either monotone increasing or monotone decreasing for all edges $(i,j)\in E$ for all times $t\geq T$.  
That is, for every edge $(i,j)$ and for all times $t\geq T$, precisely one of the following statements holds: 
    \begin{align}
        z_{ij}(t + 1)&\geq z_{ij}(t)\,; \label{eq:monotone_increase} \\
        z_{ij}(t + 1)&\leq z_{ij}(t)\,.\label{eq:monotone_decrease}
    \end{align}  
    In particular, for $(i,j)\in E$, the inequality \eqref{eq:monotone_increase} holds if $x^i = x^j$ and the inequality \eqref{eq:monotone_decrease} holds if $x^i \neq x^j$. 
    Moreover, if $x^i = x^j$, the edge weight $z_{ij}$ increases whenever $i$ and $j$ interact with each other; if $x^i \neq x^j$, the edge weight $z_{ij}$ decreases whenever $i$ and $j$ interact with each other. 
\end{lemma}

\begin{proof}
    To prove Lemma \ref{lemma:adaptive_DW_influence_T}, we separately consider situations with nodes $i$ and $j$ in different limit opinion clusters and situations with nodes $i$ and $j$ in the same limit opinion cluster. 
    
    Suppose that $i$ and $j$ are in different limit opinion clusters (i.e., $x^i \neq x^j$). 
        Consider a time $T$ such that inequalities
    \begin{align}
        |x_k (t) - x^k| &< \frac{1}{4} \min_{x^m \neq x^n} |x^m - x^n|\,; \label{eq:monotone_condition1}\\
        |x_k (t) - x_k(t')| &< \frac{\mu}{4} \min_{x^m \neq x^n} |x^m - x^n| \,; \label{eq:monotone_condition2}\\
        |x_k (t) - x^k| &< \frac{c}{2} \label{eq:monotone_condition3}
    \end{align}
    hold for all nodes $k$ and for all times $t' > t > T$. 

    Note that 
    \begin{align}
    	{\Delta^{ij} = }|x^i - x^j | \geq \min_{x^m \neq x^n} |x^m - x^n|\,.
    \label{eq:min_opinion_diff}
    \end{align}
    By \eqref{eq:monotone_condition1} and the triangle inequality,
    \begin{align*}
   	 {\Delta^{ij} = } |x^i - x^j| &= |x^i - x_i (t) + x_i (t) - x_j (t) + x_j (t) - x^j| \\
   			 &\leq |x^i - x_i(t)| + |x_i (t) - x_j (t)| + |x_j (t) - x^j| \\
   			 &< \frac{1}{2} \min_{x^m \neq x^n} |x^m - x^n| + |x_i (t) - x_j (t)|
    \end{align*}
    for all times $t\geq T$. 
    By rearranging terms and using \eqref{eq:min_opinion_diff}, we see that 
    \begin{align*}
   	 \min_{x^m \neq x^n}|x^m - x^n| &\leq {\Delta^{ij}} \\
   							 &< \frac{1}{2} \min_{x^m \neq x^n} |x^m - x^n| + |x_i (t) - x_j (t)| \,,
    \end{align*}
    which implies that
    \begin{equation} \label{eq:monotone_inequality1}
   	 \frac{1}{2} \min_{x^m \neq x^n}|x^m - x^n| < |x_i (t) - x_j (t)|\,.
    \end{equation}
    
    Suppose that $z_{ij}(t)$ increases at a time $t \geq T$ (i.e, $z_{ij}(t + 1) > z_{ij}(t))$).
  That is, nodes $i$ and $j$ are receptive at time $t$ and interact at that time. 
    Therefore,
    \[
	    x_j (t + 1) = x_j (t) + \mu(x_i(t) - x_j(t)) \,,
    \]
    which implies that 
    \begin{equation}     \label{eq:opinion_update_difference}
 	   |x_j(t + 1) - x_j (t)| = \mu|x_i(t) - x_j(t)|\,.
    \end{equation}
        By \eqref{eq:monotone_condition2}, we also have 
    \begin{align}     \label{eq:adjacent_opinion_difference}
    	|x_j(t + 1) - x_j (t)| &< \frac{\mu}{4} \min_{x^m \neq x^n} |x^m - x^n| \,.
    \end{align}
    Combining \eqref{eq:opinion_update_difference} and \eqref{eq:adjacent_opinion_difference} yields 
    \begin{equation} \label{eq:max_opinion_difference}
    	|x_i(t) - x_j(t)| < \frac{1}{4}\min_{x^m \neq x^n} |x^m - x^n| \,.
    \end{equation}

    The inequalities \eqref{eq:monotone_inequality1} and \eqref{eq:max_opinion_difference} cannot hold simultaneously, so $z_{ij}(t)$ cannot increase for any time $t \geq T$. 
    Therefore, $z_{ij}$ is monotone decreasing for $t \geq T$. Because $z_{ij}(t)$ cannot increase, nodes $i$ and $j$ are mutually unreceptive when they interact with each other. (That is, $|x_i(t) - x_j(t)|\geq c$ for any time $t \geq T$ at which $i$ and $j$ interact.)
    When $i$ and $j$ interact with each other, we update $z_{ij}$ using \eqref{eq:edge_weight_decrease}, which implies that $z_{ij}$ decreases.
    
    Now consider adjacent nodes $i$ and $j$ in the same limit opinion cluster (i.e., $x^i = x^j$). 
    Using the inequality \eqref{eq:monotone_condition3}, we have for all times $t \geq T$ that
    \begin{align}
        | x_i (t) - x_j (t) | &= | x_i (t) - x^i + x^j - x_j (t)| \notag\\ 
        				&\leq | x_i (t) - x^i | + | x^j - x_j (t) | \notag \\
        				&< \frac{c}{2} + \frac{c}{2} \notag \\
       				 &= c\,. \label{eq:same_limit_opinion}
    \end{align}

    Suppose that $i$ and $j$ interact with each other at time $t\geq T$.
    Because $|x_i (t) - x_j (t) | < c$ for all times $t\geq T$, 
    nodes $i$ and $j$ are mutually receptive and the edge weight $z_{ij}$ updates, so $z_{ij}(t + 1) = z_{ij}(t) + \gamma (1 - z_{ij}(t)) > z_{ij}(t)$. 
    If $i$ and $j$ do not interact with each other at a time $t\geq T$, then $z_{ij}(t + 1) = z_{ij}(t)$. 
    Therefore, $z_{ij}$ is monotonically increasing for times $t\geq T$ and increases when $i$ and $j$ interact with each other. 
\end{proof}

\begin{lemma} \label{lemma:edge_weight_limit}
    Consider the adaptive edge-weighted DW model with opinion update rule (\ref{eq:opinion_update_positive}, \ref{eq:opinion_update_negative}) and edge-weight update rule (\ref{eq:edge_weight_increase}, \ref{eq:edge_weight_decrease}).
Suppose that the edge weight $z_{ij}(t)$ converges to some limit for each edge $(i,j)\in E$. For all edges $(i,j)\in E$, we then almost surely have $z_{ij}(t)\to 0$ or $z_{ij}(t)\to 1$ as $t\to\infty$.
\end{lemma} 

\begin{proof}
    Suppose that $z_{ij}(t)$ converges to some limit $z^{ij}$. For $\epsilon > 0$, we choose a time $T$ such that
    \begin{align} 
  	  |z_{ij}(t) - z^{ij}| &< \frac{\epsilon}{2} \,,
    \label{eq:limit_condition1} \\
  	  | z_{ij}(t + 1) - z_{ij}(t) | &< \min\{1 - \delta, \gamma\} \frac{\epsilon}{2} 
    \label{eq:limit_condition2} 
    \end{align}
    for all times $t \geq T$. 

    Suppose that nodes $i$ and $j$ interact with each other at some time $t \geq T$. Either $| x_i(t) - x_j (t)| < c$ (i.e., $i$ and $j$ are mutually receptive) or $|x_i(t) - x_j (t) | \geq c$ (i.e., $i$ and $j$ are mutually unreceptive).

    If $| x_i(t) - x_j (t) | < c$, then $z_{ij} (t + 1) = z_{ij}(t) + \gamma (1 - z_{ij} (t))$. Therefore, $z_{ij}(t + 1) - z_{ij} (t) = \gamma( 1 - z_{ij} (t))$. 
    Using the inequality \eqref{eq:limit_condition2}, we obtain 
   \[
	    	\gamma | 1 - z_{ij} (t) | < \min\{1 - \delta, \gamma\} \frac{\epsilon}{2}\,,
\] 
which implies that  
\begin{equation} \label{eq:limit_convergence1}
	 | 1 - z_{ij}(t)| < \frac{\epsilon}{2} \,. 
\end{equation}
    Using \eqref{eq:limit_condition1} and \eqref{eq:limit_convergence1}, we obtain $0 \leq | 1 - z^{ij} (t) | \leq | 1 - z_{ij}| + | z_{ij} - z^{ij} | < \frac{\epsilon}{2} + \frac{\epsilon}{2}  = \epsilon$. Therefore, $z^{ij} = 1$. 

    If $| x_i(t) - x_j (t) | \geq c$, then $z_{ij}(t + 1) = \delta z_{ij}(t)$, which 
    implies that $z_{ij}(t + 1) - z_{ij}(t) =  (\delta - 1)z_{ij}(t)$. Using the inequality \eqref{eq:limit_condition2}, we obtain 
    \[
    	(1 - \delta)| z_{ij}(t)| < \min\{1 - \delta, \gamma\}\frac{\epsilon}{2}\,, 
\] 
which in turn implies that 
    \begin{equation} \label{eq:limit_convergence2}
    	 |z_{ij}(t)| < \frac{\epsilon}{2}  \,. 
    \end{equation} 
       Using \eqref{eq:limit_condition1} and \eqref{eq:limit_convergence2}, we obtain $0 \leq z^{ij} \leq | z^{ij} - z_{ij}(t) | + | z_{ij}(t)| < \frac{\epsilon}{2} + \frac{\epsilon}{2} = \epsilon$. Therefore, $z^{ij} = 0$.

    Nodes $i$ and $j$ almost surely interact with each other infinitely often, so we almost surely have either $z^{ij} = 0$ or $z^{ij} = 1$.
\end{proof}

\begin{theorem} \label{theorem:edge_convergence}
      Consider the adaptive edge-weighted DW model with opinion update rule (\ref{eq:opinion_update_positive}, \ref{eq:opinion_update_negative}) and edge-weight update rule (\ref{eq:edge_weight_increase}, \ref{eq:edge_weight_decrease}).
      Each  edge weight $z_{ij}(t)$ almost surely converges to $0$ or $1$. Additionally, when $z_{ij}(t)$ converges to $0$ or $1$ for each edge $(i,j)\in E$, the following statements hold:  
      \begin{itemize}
          \item{If $\lim\limits_{t\to\infty}x_i(t) = \lim\limits_{t\to\infty}x_j(t)$, then almost surely $z_{ij}(t)$ converges to $1$.}
          \item{If $\lim\limits_{t\to\infty}x_i(t) \neq \lim\limits_{t\to\infty}x_j(t)$, then almost surely $z_{ij}(t)$ converges to $0$.}
      \end{itemize}
\end{theorem}

\begin{proof} 

    We first consider the case $x^i = x^j$, and we then consider the case $x^i \neq x^j$

    Suppose that $x^i = x^j$. By Lemma \ref{lemma:adaptive_DW_influence_T}, there is a time $T_{ij}$ such that the edge weight $z_{ij}(t)$ is monotone increasing for all times $t\geq T_{ij}$ and strictly increases whenever nodes $i$ and $j$ interact with 
    each other at times $t' \geq T_{ij}$. Because $z_{ij}(t)$ is monotone increasing and takes values in the compact set $[0,1]$, it follows that $z_{ij}(t)$ converges.  By Lemma \ref{lemma:edge_weight_limit}, $z_{ij}$ almost surely converges either to $0$ or to $1$. 
    Now suppose that $z_{ij}$ converges to $0$ or $1$. By Lemma \ref{lemma:infinite_interactions}, nodes $i$ and $j$ almost surely interact infinitely often with each other, which implies that almost surely $z_{ij}$ strictly increases at infinitely many times.  
    Because the edge weight $z_{ij}(t)$ is monotone increasing and almost surely strictly increases infinitely often for times $t\geq T_{ij}$, it is guaranteed that $z_{ij}$ almost surely converges to $1$.

    Now suppose that $x^i\neq x^j$. By Lemma \ref{lemma:adaptive_DW_influence_T}, there is a time $T_{ij}$ such that the edge weight $z_{ij}(t)$ is monotone decreasing for all times $t\geq T_{ij}$ and strictly decreases whenever nodes 
    $i$ and $j$ interact with each other at times $t' \geq T_{ij}$. 
   Because $z_{ij}(t)$ is monotone decreasing and takes values in the compact set $[0,1]$, it follows that $z_{ij}(t)$ converges.  By Lemma \ref{lemma:edge_weight_limit}, $z_{ij}$ almost surely converges either to $0$ or to $1$. Now suppose that $z_{ij}$ converges to $0$ or $1$. By Lemma \ref{lemma:infinite_interactions}, nodes $i$ and $j$ almost surely interact with each other infinitely often, which implies that almost surely $z_{ij}$  strictly decreases at infinitely many times.  
    Because the edge weight $z_{ij}(t)$ is monotone decreasing and almost surely strictly decreases infinitely often for $t\geq T_{ij}$, it is guaranteed that $z_{ij}$ almost surely converges to $0$. 
\end{proof}


\subsection{Effective-graph analysis}

For a graph $G$ in the adaptive edge-weighted DW model with opinion update rule (\ref{eq:opinion_update_positive}, \ref{eq:opinion_update_negative}) and edge-weight update rule 
(\ref{eq:edge_weight_increase}, \ref{eq:edge_weight_decrease}),
the \emph{effective graph} $G_\mathrm{eff}(t)$ is the time-dependent subgraph of $G$ that includes only the edges of $G$ between nodes that are receptive to each other (i.e., between nodes with opinions that differ by less than the confidence bound). That is, 
\begin{align*}
    G_\mathrm{eff}(t) = (V,E_\mathrm{eff}(t))\,, 
\end{align*}
where $E_\mathrm{eff}(t) = \{(i,j) \in E \text{ such that } |x_i(t) - x_j(t)| < c\}$.
If $G_\mathrm{eff}(t)$ is eventually constant with respect to time (i.e., there is a time $T$ such that $G_\mathrm{eff}(t) = G_\mathrm{eff}(T)$ for all times $t \geq T$), then $G_\mathrm{eff}(T)$ is {the} \emph{limit effective graph}.

{We now prove Theorem \ref{thm:effgraph}, which (1) guarantees that a limit effective graph exists whenever limit effective opinions do not differ by exactly the confidence bound $c$ and (2) specifies two important properties of limit effective graphs. 
In practice, it is very unlikely for two limit opinions to differ by exactly $c$.} Accordingly, limit effective graphs always exist except for this pathological situation.

\begin{theorem}\label{thm:effgraph}
Consider the adaptive edge-weighted DW model with opinion update rule (\ref{eq:opinion_update_positive}, \ref{eq:opinion_update_negative}) and edge-weight update rule (\ref{eq:edge_weight_increase}, \ref{eq:edge_weight_decrease}). 
Suppose that no two limit opinions $x^i$ and $x^j$ differ by exactly $c$ (i.e., {$\Delta^{ij} := |x^i - x^j| \neq c$} for all nodes $i$ and $j$).  
In a simulation of the model, it is then the case that the effective graph $G_\mathrm{eff}(t)$ is eventually constant with respect to time (i.e., there is a limit effective graph). 
Now suppose that {the} limit effective graph exists for a simulation.
If adjacent nodes $i$ and $j$ have the same limit opinion (i.e., $x^i = x^j$), then the edge $(i,j)$ is in the limit effective graph. 
Additionally, if the edge $(i,j)$ is in the limit effective graph, then almost surely $x^i = x^j$.
\end{theorem}

To prove Theorem \ref{thm:effgraph}, we first prove two lemmas. 

\begin{lemma}\label{lemma:baseDW_eff_edges}
Consider the adaptive edge-weighted DW model with opinion update rule (\ref{eq:opinion_update_positive}, \ref{eq:opinion_update_negative}) and edge-weight update rule (\ref{eq:edge_weight_increase}, \ref{eq:edge_weight_decrease}).
In a simulation, there is a time $T$ such that the following statements hold for all adjacent nodes $i$ and $j$.
\begin{itemize}
    \item[(1)]{If $\Delta^{ij} = |x^i - x^j| < c$, then $|x_i(t) - x_j(t)| < c$ and the edge $(i,j)$ is in the effective graph for all times $t \geq T$.}
    \item[(2)]{If $\Delta^{ij} = |x^i - x^j| > c$, then $|x_i(t) - x_j(t)| > c$ and the edge $(i,j)$
    is not in the effective graph for any time $t \geq T$.}
\end{itemize}
\end{lemma}

\begin{proof}
Consider adjacent nodes $i$ and $j$ with
$\Delta^{ij} \neq c$. Choose a time $T_{ij}$ such that
\begin{equation}
    |x_k(t) - x^k| < \frac{1}{2} |c - \Delta^{ij}|
\end{equation}
for nodes $k \in \{i,j\}$ and all times $t \geq T_{ij}$.

First suppose that $\Delta^{ij} < c$. For all times $t\geq T_{ij}$, we have
\begin{align*}
	|x_i(t) - x_j(t)| &\leq |x_i(t) - x^i| + |x^i - x^j| + |x_j(t) - x^j| \\
			&< 2\left(\frac{1}{2}\right)(c - \Delta^{ij}) + \Delta^{ij} \\
			&= c \,.
\end{align*}
Therefore, the edge $(i,j)$ is in the effective graph for all times $t \geq T_{ij}$.

Now suppose that $\Delta^{ij} > c$. Without loss of generality, let $x^i > x^j$. For all times $t\geq T_{ij}$, we have
\begin{align*}
x_i(t) - x_j(t) &> \left(x^i -  \frac{1}{2} |c - \Delta^{ij}| \right) - \left(x^j+ \frac{1}{2} |c - \Delta^{ij}| \right) \\
		&= (x^i - x^j) - |c - \Delta^{ij}|\\
		&= \Delta^{ij} - \Delta^{ij} + c \\
		&= c \,.
\end{align*}
Therefore, the edge $(i,j)$ is not in the effective graph for any time $t \geq T_{ij}$.

With
\begin{equation}
	T = \max\limits_{(i,j) \in E} \{T_{ij} \text{~such that~} |x^i - x^j| \neq c\} \,,
\end{equation} it follows that statements (1) and (2) hold for all times $t \geq T$.

\vspace{-52pt}
\[\] 

\end{proof}

\begin{lemma}\label{lemma:baseDW_less_than_c}
For adjacent nodes $i$ and $j$ with $|x^i - x^j| < c$, we have $x^i = x^j$ almost surely.
\end{lemma}

\begin{proof}
Fix adjacent nodes $i$ and $j$ with $\Delta^{ij} = |x^i - x^j| < c$.
Without loss of generality, let $x^i > x^j$.
We then have $0 < \Delta^{ij} < c$. 

Fix $\epsilon$ so that $0 < \epsilon < \min\{\frac{1}{4} (c - \Delta^{ij}), \frac{\Delta^{ij}}{2(1+ 1/\mu)}\}$,
and choose a time $T_{ij}$ so that
\begin{equation}\label{eq:baseDW_less_c_ineq1}
	|x^k - x_k(t)| < \epsilon
\end{equation}
for each node $k$ and all times $t \geq T_{ij}$. By Lemma \ref{lemma:infinite_interactions}, there is almost surely some time $t \geq T_{ij}$ at which nodes $i$ and $j$ interact with each other.
The inequality $\epsilon < \frac{1}{4} (c - \Delta^{ij})$ implies that
\begin{align*}
    |x_i(t) - x_j(t)| &\leq |x_i(t) - x^i| + |x^i - x^j| + |x^j - x_j(t)| \\
    			&< \frac{1}{4} (c - \Delta^{ij}) + \Delta^{ij} + \frac{1}{4} (c - \Delta^{ij}) \\
   			 &= \frac{1}{2} (\Delta^{ij} + c) < c \,,
\end{align*}
so nodes $i$ and $j$ are receptive to each other when they interact with each other at time $t$. Therefore, they update their opinions and
\begin{align}
    x_j(t + 1) &= x_j(t) + \mu (x_i(t) - x_j(t)) \nonumber \\
    		&\geq x_j(t) + \mu (x^i - \epsilon - (x^j + \epsilon)) \nonumber \\ 
    		&= x_j(t) + \mu (\Delta^{ij} - 2\epsilon) \nonumber \\
    		&\geq (x^j - \epsilon) + \mu (\Delta^{ij} - 2\epsilon) \nonumber \\
    		&> x^j + \epsilon \,, \label{eq:baseDW_less_c_ineq2}
\end{align}
where the last inequality holds because 
$\epsilon < \frac{\Delta^{ij}}{2(1 + 1/\mu)}$, which we rearrange to obtain 
$2\epsilon < \mu(\Delta^{ij} - 2\epsilon)$.

The inequality \eqref{eq:baseDW_less_c_ineq1} guarantees that
\begin{equation}
	|x^j - x_j(t + 1)| < \epsilon \,,
\end{equation}
which cannot hold simultaneously with the inequality \eqref{eq:baseDW_less_c_ineq2}.
Therefore, nodes $i$ and $j$ almost surely do not interact with each other at times $t \geq T_{ij}$, which implies that
$0 < x^i - x^j < c$ with probability $0$. 
Because $x^i - x^j < c$ by assumption, we almost surely have $x^i = x^j$.

\vspace{-52pt}
\[\] 

\end{proof}

We now use Lemma \ref{lemma:baseDW_eff_edges} and Lemma \ref{lemma:baseDW_less_than_c} to prove Theorem \ref{thm:effgraph}.

\begin{proof}[Proof of Theorem \ref{thm:effgraph}]

Consider one simulation of our adaptive edge-weighted DW model. By Lemma \ref{lemma:baseDW_eff_edges}, there is a time $T$ such that statements (1) and (2) of Lemma \ref{lemma:baseDW_eff_edges} hold for all times $t \geq T$.
If we suppose that no two limit opinions $x^i$ and $x^j$ satisfy $|x^i - x^j| = c$, then $E_\mathrm{eff}(t) = E_\mathrm{eff}(T)$ for all times $t\geq T$. 
Therefore, the effective graph is eventually constant for all $t \geq T$. 
That is, {the limit effective graph exists.}

Now suppose that {the} limit effective graph exists.

For adjacent nodes $i$ and $j$ with the same limit opinion (i.e., $x^i = x^j$), we know that $|x^i - x^j| = 0 < c$. By statement (1) in Lemma \ref{lemma:baseDW_eff_edges}, it follows that there exists a time $T$ such that edge $(i,j)$ is in the effective graph for 
all times $t \geq T$. Therefore, the edge $(i,j)$ is in the limit effective graph.

Now consider edge $(i,j)$ in the limit effective graph. We seek to show that $x^i = x^j$.
Because $(i,j)$ is in the limit effective graph, there exists a time $\tilde{T}$ such that $(i,j) \in E_\mathrm{eff}(t)$ for all times $t \geq \tilde{T}$.
Consequently, by statement (2) of Lemma \ref{lemma:baseDW_eff_edges}, it cannot be the case that $|x^i - x^j| > c$. 
Therefore, $|x^i - x^j| < c$.
By Lemma \ref{lemma:baseDW_less_than_c}, we almost surely have $x^i = x^j$. 
\end{proof}


{\section{Details of our Numerical simulations}}
\label{sec:num_simulations}


\subsection{Network structures}\label{sec:network_structures}

We simulate our adaptive edge-weighted DW model on both synthetic {graphs} and {graphs that arise from real-world social networks}.
We employ complete graphs, Erd\H{o}s--R\'{e}nyi (ER) graphs, degree-regular cycle-like graphs, and {several real-world networks}.
{The real-world networks that we consider are the \textsc{NetScience} network \cite{Netscience} (which is a network of network-science researchers and the coauthorship relationships between them) and the \textsc{Facebook100} networks (which are single-university Facebook friendship networks)~\cite{Facebook100}. The {graphs} in our simulations have different edge densities
\begin{equation}
	\rho = \frac{2m}{N(N - 1)} \,, \label{eq:edge_density}
\end{equation}
{where we recall that $m = |E|$ is the number of edges in a graph and that $N = |V|$ is the number of nodes in a graph}.
A $G(N,p)$ ER graph has $N$ nodes, and there is an edge between each two nodes with independent, homogeneous probability $p$.
{In a degree-regular cycle-like graph, each node has the same degree $k$ (which is even) and is adjacent to $k/2$ neighbors on each side.}
{A $k$-regular cycle-like graph {with} $N$ nodes has adjacency-matrix entries
\begin{equation}
	A_{ij} = 
		\begin{cases}
			   1\,, & 1 \leq |i - j| \!\!\!\! \mod {N} \leq k/2  \\
			     0\,, & \textrm{otherwise}\,.
		 \end{cases}
\end{equation}		 
}

In Table \ref{table:Table0}, we summarize the details of our synthetic graphs.  
For our simulations on complete graphs, we consider graph sizes (i.e., the numbers of nodes) $N \in \{100, 200, 1000\}$. 
For our simulations on ER graphs, we consider graph sizes $N \in \{100,1000\}$ and connection probabilities $p \in \{0.01, 0.1\}$ to investigate the effect of edge density on the dynamics of our adaptive edge-weighted DW model.  
For our simulations on degree-regular cycle-like graphs, we consider degrees $k \in \{10, 50, 100, 200\}$ to investigate the effect of mean degree on our model's dynamics. 
All of our $k$-regular cycle-like graphs have $N = 1000$ nodes.

In Table \ref{table:Table1}, we indicate the numbers of nodes and edges of the real-world networks that we examine.
For each {graph}, we use the largest connected component (i.e., the component with the most nodes).
We consider two types of real-world networks. The \textsc{NetScience} network is a graph of network-science researchers and coauthorship relationships between them~\cite{Netscience},
and each \textsc{Facebook100} network is a graph of individuals 
and the Facebook ``friendships" between them at one university at a single point in time in fall 2005~\cite{Facebook100}.

\begin{table}[h]
\centering
\small 
\setlength{\tabcolsep}{3.5pt}
\caption{The synthetic graphs on which we simulate our adaptive edge-weighted DW model.}
\label{table:Table0}
\begin{tabular}{|p{0.3 \linewidth}|p{0.6\linewidth}|} 
\hline
{Graph} type & Parameters \\ 
\hline \hline
Complete graph & Sizes $N \in \{100,200,1000\}$ \\ 
\hline
$G(N,p)$ ER graph & Sizes $N \in \{100,1000\}$; connection probabilities $p \in \{0.01, 0.1\}$ \\ 
\hline
$k$-regular cycle-like graph & Size $N = 1000$; degrees $k \in \{10, 50, 100, 200\}$ \\
\hline
\end{tabular}
\end{table}

\begin{table}[h]
\centering
\small
\setlength{\tabcolsep}{4pt}
\caption{The real-world networks on which we simulate our adaptive edge-weighted DW model. For each {graph}, we use the largest connected component and indicate the numbers of nodes and edges in that component.
}
\label{table:Table1}
\begin{tabular}{|l|r|r|} 
\hline
{Graph} & \multicolumn{1}{c|}{Number of nodes} & \multicolumn{1}{c|}{Number of edges} \\ 
\hline \hline
{\sc NetScience} & 379 & 914 \\ 
\hline 
{\sc Caltech} & 762 & 16,651 \\
\hline
{\sc Reed} & 962 & 18,812 \\ 
\hline 
{\sc Swarthmore} & 1,657 & 61,049 \\ 
\hline 
{\sc Smith} & 2,970 & 97,133 \\ 
\hline
\end{tabular}
\end{table}

We run numerical simulations on both dense graphs and sparse graphs. 
We consider the complete graphs, $G(1000, 0.1)$ ER graphs, and $k$-regular cycle-like graphs with $k \in \{50, 100, 200\}$ to be ``dense". We consider the $G(1000, 0.01)$ ER graphs, the 10-regular cycle-like graph, the \textsc{NetScience} network, and the \textsc{Facebook100} networks to be ``sparse". The graphs that we call ``dense" have higher edge densities than the graphs that we call ``sparse". Our usage of these terms deviates from standard conventions \cite{newman2018}, and it is motivated by two differences between the trends that we observe in our numerical simulations on these two sets of graphs.


\subsection{{Simulation specifications and guarantees}} 
\label{subsec: simulation_specifications}

The parameters of our adaptive edge-weighted DW model are the edge-weight increase parameter $\gamma$, the edge-weight decrease parameter $\delta$, the confidence bound $c$, the initial edge-weight parameter $z_0$, and the compromise parameter $\mu$. The pair $(\gamma, \delta) = (0,1)$ corresponds to the baseline DW model. In Table \ref{table:Table2}, we summarize the parameter values that we consider in our numerical simulations.

\begin{table}[h]
\centering
\small
\setlength{\tabcolsep}{4pt}
\caption{The parameter values for our numerical simulations of our adaptive edge-weighted DW model.}
\label{table:Table2}
\begin{tabular}{|l|l|} 
\hline
Parameter & Values \\ 
\hline\hline
$\gamma$ & $\{0.1, 0.5, 0.9\}$ \\ 
\hline
$\delta$ & $\{0.1, 0.5, 0.9\}$ \\ 
\hline
$c$ & $\{0.1, 0.2, 0.3, 0.4, 0.5, 0.6, 0.7, 0.8, 0.9\}$ \\
\hline
$z_0$ & $\{0.1, 0.5, 0.9\}$ \\
\hline
$\mu$ & $\{0.1, 0.3, 0.5\}$ \\
\hline
\end{tabular}
\end{table}

Our simulations include randomness from the initial node opinions, the selection of a node pair to interact at each discrete time step, and the generation of a graph for our ER graphs. 
We use Monte Carlo simulations to account for this randomness.
For the simulations of our adaptive edge-weighted DW model on ER graphs, we generate 5 graphs for each parameter set and pair $(N,p)$. 
For each synthetic graph, we generate 20 sets of initial opinions; for each real-world network, we generate 10 sets of initial opinions. 
We calculate means of all summary statistics for the 10 or 20 simulations from the sets of initial opinions.

The opinions of the nodes of a graph can take an arbitrarily long time to converge to limit opinions, so we specify a stopping criterion to determine when we stop a simulation. 
We base our stopping criterion on the stopping criterion in \cite{bcm_paper}.
At each discrete time $t$, we consider the effective graph $G_{\mathrm{eff}}(t)$. 
Let $K_1(t),\ldots, K_r(t)$ be the sets of nodes of the connected components of $G_{\mathrm{eff}}(t)$; we refer to these sets as \emph{opinion clusters} at time $t$. 
We stop a simulation if the maximum opinion difference between the nodes in each opinion cluster is less than some specified tolerance.
That is, we check 
\begin{equation} \label{stopping}
	\max \{|x_i(t) - x_j(t)| : i,j \in K_r(t) \text{ for some } r \} < \mathrm{tol} \,, 
\end{equation}
where $\mathrm{tol}$ denotes the tolerance. 
We use $\mathrm{tol} = 1 \times 10^{-2}$ in all of our simulations. 
If a simulation reaches the stopping criterion \eqref{stopping} at time $T_f$, we refer to $T_f$ as the simulation's ``convergence time" and to the effective graph at $T_f$ as the ``final effective graph". 
We refer to the opinion clusters of a final effective graph as the ``final opinion clusters".

{If a simulation takes too long to converge, we stop it when it reaches a bailout time. 
For the 10-regular cycle-like graph, the bailout time is $10^{10}$. For all other graphs, the bailout time is $10^8$. 
 If a simulation reaches the bailout time, we say that it ``bails out'', and we stop it whether or not it meets the stopping criterion \eqref{stopping}. With our bailout times, most
 of our simulations do not bail out}. 
In our numerical experiments, the $G(1000, 0.01)$ ER graphs had the largest percentage of bailouts, with 0.22\% of all simulations bailing out and 6.56\% of the simulations bailing our when $\delta = \gamma = c = 0.1$. In these ER graphs, bailouts occurred when $\delta = 0.1$ and $c \in \{0.1, 0.2\}$.

Our theoretical results about effective graphs help inform our stopping criterion \eqref{stopping}.
In Theorem \ref{thm:effgraph}, we proved that an effective graph almost surely eventually has edges only between nodes in the same limit opinion cluster.
Therefore, using the stopping criterion \eqref{stopping}, the final opinion clusters in our simulations approximate the limit opinion clusters.

In our simulations of our adaptive edge-weighted DW model on complete graphs, the {sets of nodes in the} final opinion clusters are precisely the limit opinion clusters if the tolerance is less than the confidence bound.
In Theorem \ref{thm:final_op_clusters}, we give a precise statement of this result.

\begin{theorem} \label{thm:final_op_clusters}
    Let $G = (V,E)$ be a complete graph. Suppose that $\mathrm{tol} < c$, and let $T_f$ be the time that a simulation reaches the stopping criterion \eqref{stopping}. 
    Let $\{S_a\}_{a = 1}^r$ be the set of final opinion clusters, and let $J(i, x_i(t))$ denote the set of nodes that are receptive to node $i$ at time $t$ (i.e., the nodes $j$ such that $|x_i(t) - x_j(t)| < c$). 
    {The following statements hold.
    \begin{enumerate}
        \item[(1)]{The inequality $\min_{i \in S_a, j \in S_b, a\neq b} \{|x_i(t) - x_j(t)|\} \geq c$ holds for all times $t \geq T_f$.} 
        \item[(2)]{For a fixed final opinion cluster $S_a$ and node $i \in S_a$, we have $J(i, x_i(t)) = S_a$ for all times $t \geq T_f$.}
    \end{enumerate}} 
\end{theorem}

\begin{proof}
    We proceed by induction on time $t \geq T_f$. 

    We start with the base case $t = T_f$. For any edge $(i,j)$, we have $|x_i - x_j| < c$ if and only if $(i,j) \in E_{\mathrm{eff}}(t)$ (i.e., $i$ and $j$ are in the same final opinion cluster $S_a$). 
    In a complete graph, $(i,j) \in E$ for all distinct nodes $i$ and $j$. 
    Therefore, if nodes $i$ and $j$ are not in the same opinion cluster at time $t$, it is necessarily true that $|x_i(t) - x_j(t)| \geq c$. 
    This verifies statement {(1).}
    To verify statement {(2)},
    fix node $i \in S_a$. For any node $j \in S_a$, we have $|x_i(t) - x_j(t)| < \mathrm{tol} < c$, which implies that $j \in J(i,x(t))$. Moreover, for $j \in J(i, x(t))$, we have $|x_i(t) - x_j(t)| < c$, so nodes $i$ and $j$ are in the same opinion cluster at time $t$ (i.e., $j \in S_a$). 
    Consequently, $S_a = J(i, x(t))$. 

    {We now do the inductive step.}
        Suppose that statements
    {(1) and (2)}
    hold for times $t \geq T_f$.

    Fix a node $i \in S_a$.
    Node $i$ updates its opinion at time $t + 1$ if and only if it interacts with a node $j \in S_a$ at time $t + 1$. 
    Therefore, for any final opinion cluster $S_a$, we have 
    \begin{align}
        \min_{i \in S_a} x_i(t + 1) &\geq  \min_{i \in S_a} x_i(t) \label{eq:num1} \,, \\
        \max_{i \in S_a} x_i(t + 1) &\leq \max_{i\in S_a} x_i(t) \label{eq:num2}\,.
    \end{align}
 By invoking \eqref{eq:num1} and \eqref{eq:num2}, we verify statement {(1)}:
 \begin{align}
     \min_{\substack{\ i \in S_a  \\ \ j \in S_b  \\ a \neq b}} \; |x_i(t + 1) - x_j(t + 1)|  &\geq 
      \min_{\substack{\ i \in S_a  \\ \  j \in S_b  \\ a \neq b}} \; |x_i (t) - x_j(t)| \nonumber \\  &\geq c \,. \label{eq:num3}
 \end{align}

 Now fix a final opinion cluster $S_a$. For nodes $i \in S_a$ and $j \in S_b$ in distinct final opinion clusters (i.e., $a \neq b$), the inequality \eqref{eq:num3} implies that
 \begin{align}
	  |x_i(t + 1) - x_j(t + 1)| &\geq 
      \min_{\substack{\ i \in S_a \\ \  j \in S_b \\ a \neq b}} \; 
      |x_i(t + 1) - x_j(t + 1)|   \notag \\ 
	 					 &\geq c \,,
 \end{align}
 which implies that $j \notin J(i,x(t + 1))$ and hence that $J(i,x(t + 1)) \subseteq S_a$. 
 Using the inequalities \eqref{eq:num1} and \eqref{eq:num2} yields
 \begin{align}
     \max_{i \in S_a} x_i(t + 1) - \min_{i \in S_a} x_i(t + 1) 
     	&\leq \max_{i \in S_a} x_i(t) - \min_{i \in S_a} x_i(t) \nonumber \\
     	&\leq \mathrm{tol} < c \,,
 \end{align}
which implies that $S_a \subseteq J(i,x(t + 1))$, verifying statement {(2).}
\end{proof}

In our simulations of our adaptive edge-weighted DW model, we set the tolerance to be less than the confidence bound. 
Therefore, by Theorem \ref{thm:final_op_clusters}, the 
final opinion clusters in our simulations on complete graphs are the same as 
the limit opinion clusters. 
Theorem \ref{thm:final_op_clusters} also provides heuristics that suggest that the final opinion clusters are reasonable approximations of the limit opinion clusters for our simulations on other graph structures.

{Our definition of limit opinion clusters does not require the limit opinion clusters to be connected, whereas our final opinion clusters (which we define using effective graphs) are always connected. 
For simulations on graphs other than complete graphs, the sets of nodes in the final and limit opinion clusters are not necessarily the same. 
In particular, for non-complete graphs, it is possible for two disconnected sets of nodes to converge to the same limit opinion, which places them in the same limit opinion cluster but in different final opinion clusters. However, in practice, 
it is extremely unlikely for disconnected sets of nodes to converge to the same limit opinion.}


\subsection{Quantifying model behaviors} 

In our simulations, we examine the convergence time, the number of opinion clusters, and the Shannon entropy. As we stated in Section \ref{subsec: simulation_specifications}, the
convergence time is the number $T_f$ of time steps at which a simulation reaches the stopping criterion \eqref{stopping}.
When we consider final opinion clusters, we distinguish between ``major" and ``minor" opinion clusters.
We define a major opinion cluster as a final opinion cluster with more than $1\%$ of the nodes of a {graph}. Accordingly, a major opinion cluster needs to have at least 2 nodes in a 100-node graph and at least 11 nodes in a 1000-node graph.
A minor opinion cluster is a final opinion cluster that is not a major opinion cluster. 
{Our designation of major and minor opinion clusters follows the choices in \cite{laguna2004, bcm_paper}.}

The numbers of major and minor opinion clusters help characterize whether a simulation results in opinion ``consensus" or ``fragmentation".\footnote{In studies of opinion dynamics, it is common to refer to a state with three of more opinion clusters as ``fragmented" and to a state with exactly two opinion clusters as ``polarized"~\cite{starnini2025}. Following \cite{bcm_paper}, we find it more convenient to use the term ``fragmented" for any {non-consensus state} at the convergence time.}
Intuitively, it seems sensible to view a situation with a 1000-node {graph} with 995 nodes with one limit opinion and 5 nodes with some other limit opinion as a consensus state. 
Accordingly, we say that a simulation is in a ``consensus" state if it has only one major opinion cluster, even if there are also 
minor opinion clusters. 
A simulation that results in at least two major opinion clusters is in a state of ``opinion fragmentation". 
Two states of opinion fragmentation can have wildly different numbers of major and minor opinion clusters.

To more precisely quantify opinion fragmentation, we 
account for the number and sizes of opinion clusters. 
Following \cite{entropy,bcm_paper}, we do this by calculating Shannon entropy.
At time $t$, suppose that there are $R$ opinion clusters $K_r(t)$, where $r \in \{1,2, \ldots, R\}$.
The set $\{K_r(t)\}_{r = 1}^R$ of opinion clusters is a partition of the set of nodes of a graph. 
The fraction of nodes in the $r$th opinion cluster is 
$|K_r(t)|/N$, where $N$ is the number of nodes.
The Shannon entropy is
\begin{equation}\label{eq:shannon_entropy}
	H(t) = -\sum_{r = 1}^R \frac{|K_r(t)|}{N}\ln \left(  \frac{|K_r(t)|}{N} \right)\,. 
\end{equation}	

For each simulation, we calculate the Shannon entropy $H(T_f)$ at the convergence time. Shannon entropy is larger when there are more opinion clusters or, for a fixed number of opinion clusters, when the opinion clusters have similar sizes 
rather than more heterogeneous sizes. A larger Shannon entropy thus indicates greater opinion fragmentation. 


\subsection{Summary of our simulation results}\label{sec:sim_results}

We now detail the results of our numerical simulations of our adaptive edge-weighted DW model. 
We organize our results by {graph} type and discuss them in detail in Section \ref{details}. 

In our numerical simulations, we examine the dependence of the number of major opinion clusters, the Shannon entropy (see Eq.~\eqref{eq:shannon_entropy}), and the convergence time on the confidence bound $c$.  
We consider several values of the edge-weight increase parameter $\gamma$, the edge-weight decrease parameter $\delta$, the initial edge-weight parameter $z_0$, and the compromise parameter $\mu$ (see Table \ref{table:Table2}). 
Given a {graph structure}, each of our plots summarizes our simulations for a pair of parameter values --- either $(\gamma,\delta)$ with a common $z_0$ value or $(z_0,\delta)$ with a common $\gamma$ value --- and each curve in a plot corresponds to simulations for a fixed compromise parameter $\mu\in \{0.1, 0.3, 0.5\}$. 
Each point on a curve is a mean of the results of our numerical simulations for the corresponding set of parameters (see Section \ref{subsec: simulation_specifications}).
We show the standard deviations of these quantities as error bars. 
 
For most of our {graphs}, {when the confidence bound} $c$ is small, we usually obtain fewer major opinion clusters in our adaptive edge-weighted DW model for $\mu = 0.1$ than for larger values of $\mu$. 
However, this is not the case for the \textsc{NetScience} network (see Section \ref{sec:netscience}).
We do not observe this difference in the baseline DW model, where the numbers of major opinion clusters tend to be similar across all $\mu$ values. 
For some {graphs}, such as ER graphs and $k$-regular cycle-like graphs, 
{we observe that} the number of major opinion clusters decreases as we decrease $\delta$ when $\mu = 0.1$ and $c$ is small.

The trends {in the Shannon entropies} in our simulations {are similar to the}
trends in the number of major opinion clusters, except for the \textsc{Facebook100} networks, the \textsc{NetScience} network, and the 10-regular cycle-like graph. 
We discuss the differences in these trends in Section \ref{details}.

{We observe two distinct trends in the convergence-time properties
of our adaptive edge-weighted DW model.  
trends appear to differ between our dense graphs (i.e., the complete graphs, $G(1000, 0.1)$ ER graphs, and $k$-regular cycle-like graphs with $k \in \{50,100,200\}$) and our sparse graphs (i.e., the $G(1000,0.01)$ ER graphs, the 10-regular cycle-like graph, the \textsc{NetScience} network, and the \textsc{Facebook100} networks).\footnote{See Section \ref{sec:network_structures} for our definitions of ``dense" and ``sparse", which are not standard.}
 In our {dense graphs}, our adaptive edge-weighted DW model tends to have smaller convergence times than the baseline DW model for small confidence bounds $c$. (For example, $c = 0.1$ and $c = 0.2$ for the 50-regular cycle-like graph and $c = 0.1$ for the complete and $G(1000,0.1)$ ER graphs.)
 In our {sparse graphs}, when the edge-weight decrease parameter is $\delta = 0.1$, our adaptive edge-weighted DW model tends to have larger variances in the convergence time for all values of $c$ and longer convergence times than the baseline DW model for small and moderate values of $c$.
Additionally, the convergence time {for our sparse graphs} decreases as we increase $\delta$.
For $\delta = 0.9$, {our simulations of our adaptive edge-weighted DW model on our sparse graphs converge faster than our simulations of the baseline DW model on these graphs}.
}

{We hypothesize that a {graph's} edge density is an important determining factor for the convergence-time properties of our adaptive edge-weighted DW model. For {graphs} of a fixed size $N$, there are fewer edges to choose for interactions in a {sparse graph} than in a {dense graph}. 
Therefore, in most situations, we expect the total edge weight 
to be smaller in {sparse graphs} than in {dense graphs}.
 The probability to select a specific edge is its weight divided by the total edge weight, so we expect it to be more likely to
 select a small-weight edge in a {sparse graph} than in a {dense graph}.
{When two nodes are linked by an edge whose weight is smaller than 0.1 (which is the smallest value of $z_0$ that we consider), these two nodes necessarily have interacted previously and been unreceptive to each other. Accordingly, when they interact with each other again, it is likely that they will again be} unreceptive to each other.
 Because such nodes are likely to consistently be unreceptive to each other, we expect the probability to choose unreceptive neighbors for interaction to be larger in sparse graphs than in dense graphs.
 Consequently, in our adaptive edge-weighted DW model, the convergence time and its variance 
 are larger in sparse graphs than in dense graphs.}


\section{Simulation results} \label{details}

\subsection{Complete graphs}
\label{sec:complete}

\begin{figure*}[htbp]
  \includegraphics[width=\linewidth]{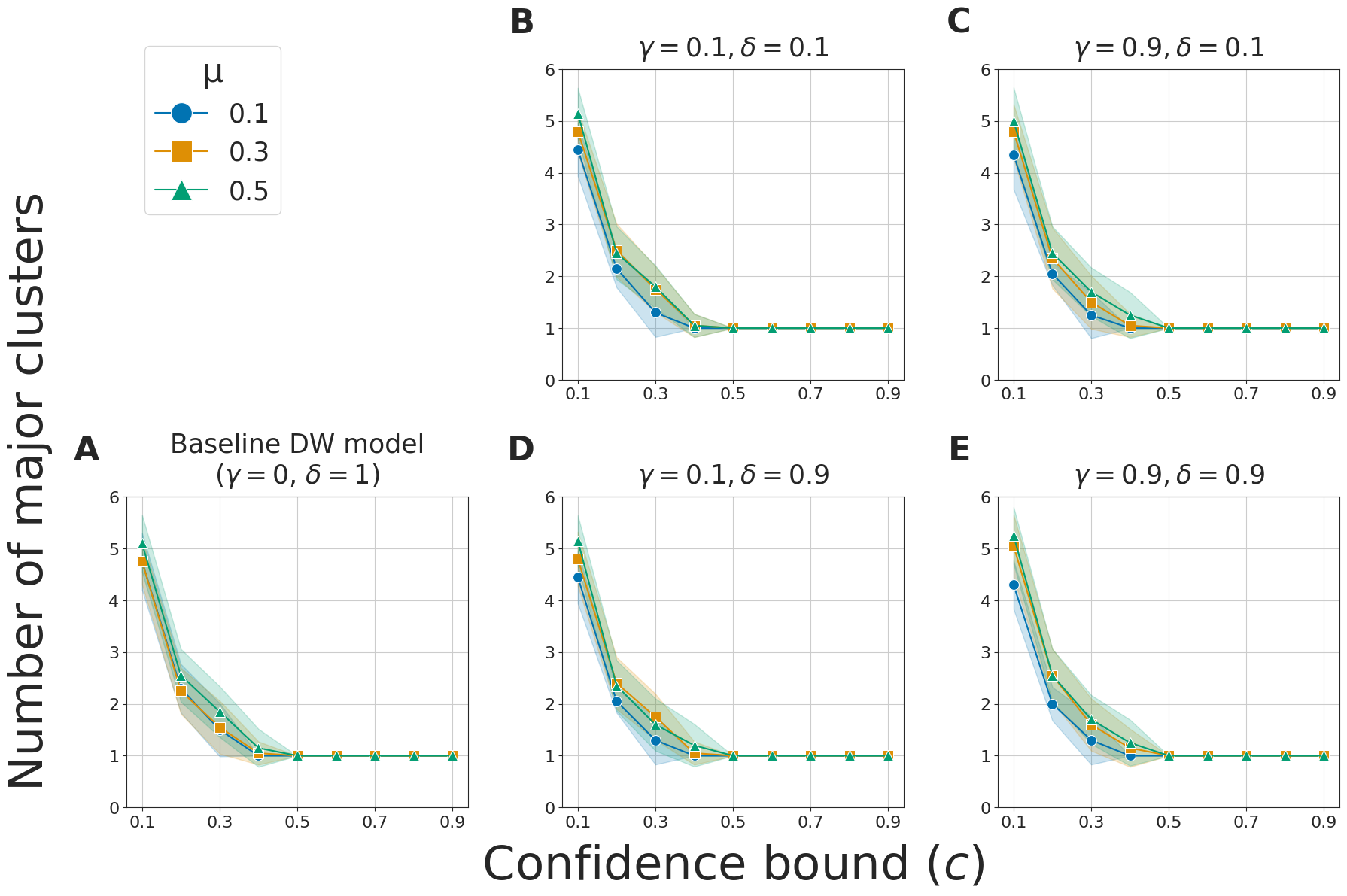}
  \caption{The numbers of major opinion clusters 
  in simulations of our adaptive edge-weighted DW model versus the confidence bound $c$ for different values of the edge-weight increase parameter $\gamma$ and the edge-weight decrease parameter $\delta$ on complete graphs with $N = 100$ nodes. The initial edge-weight parameter is $z_0 = 0.1$.
  }
  \label{figure:complete_opinions}
\end{figure*}

We simulate our adaptive edge-weighted DW model on complete graphs with $N \in \{100,200,1000\}$ nodes. 
For a parameter set $(\gamma, \delta, c, \mu, z_0)$, we obtain a similar number of major opinion clusters for all of these graphs.
For $N \in \{100,200\}$ nodes and small values of the confidence bound $c$ (specifically, $c \leq 0.2$), our simulations yield a smaller number of major opinion clusters for $\mu = 0.1$ than for 
$\mu = 0.3$ and $\mu = 0.5$ (see Figure \ref{figure:complete_opinions}). 
This differs from our observations for the baseline DW model, for which we observe similar numbers of major opinion clusters for all values of $\mu$.

\begin{figure*}[htbp]
  \includegraphics[width=\linewidth]{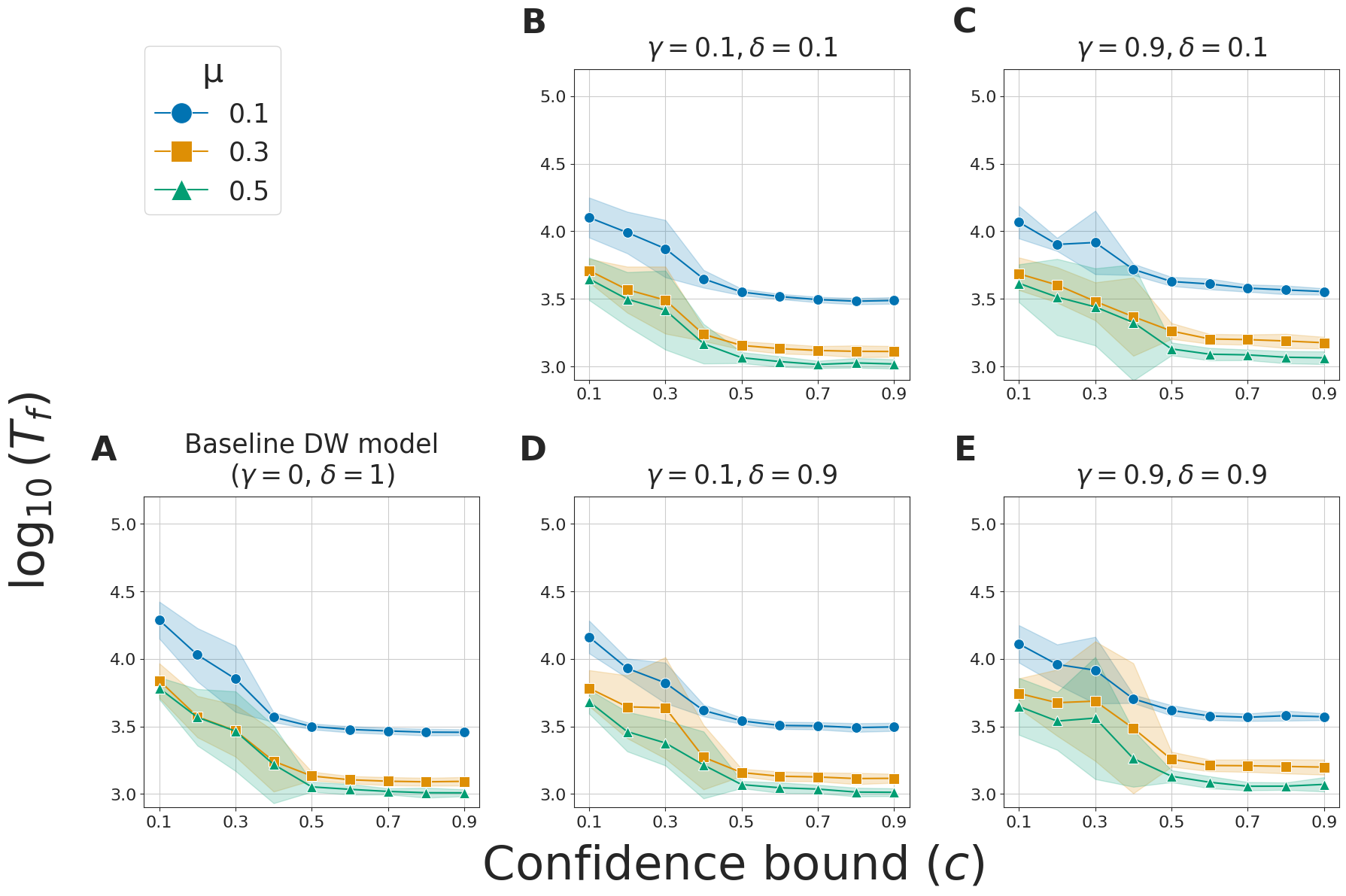}
  \caption{The logarithm of the convergence times of simulations of our adaptive edge-weighted DW model versus the confidence bound $c$ for different values of the edge-weight increase parameter $\gamma$ and the edge-weight decrease parameter 
  $\delta$ on complete graphs with $N = 100$ nodes. The initial edge-weight parameter is $z_0 = 0.1$.
  }
  \label{figure:complete_time}
\end{figure*}

For small confidence bounds $c$ and graph sizes $N \in \{100,200\}$, our adaptive edge-weighted DW model usually converges faster
than the baseline DW model (see Figure \ref{figure:complete_time}).
However, we do not observe this disparity when $(z_0,\delta) = (0.9, 0.9)$.
For $N = 1000$, our adaptive edge-weighted DW model and the baseline DW model have similar numbers of major opinion clusters.
Additionally, for both our adaptive edge-weighted DW model and the baseline DW model, we observe a larger convergence-time variance for the complete graph with $N = 1000$ nodes than for 
the complete graphs with $N \in \{100,200\}$ nodes.

For all examined values of $(\gamma,\delta)$ and all graph sizes $N$, our edge-weighted DW model and the baseline DW model have similar Shannon entropies.
This suggests that our model's adaptive edge weights do not have a significant effect on opinion fragmentation for complete graphs.


\subsection{Erd\H{o}s--R\'{e}nyi graphs} \label{sec:ER}

We simulate our adaptive edge-weighted DW model on $G(N,p)$ ER graphs with $(N,p) \in \{(1000, 0.1), (1000, 0.01), (100, 0.1)\}$.  

Our simulations on $G(1000, 0.1)$ ER graphs yield more major opinion clusters than our simulations on $G(1000, 0.01)$ ER graphs. 
For both $G(1000, 0.1)$ ER graphs and $G(1000, 0.01)$ ER graphs, when the confidence bound $c$ is small (specifically, when $c = 0.1$ or $c \in \{0.1, 0.2\}$), we observe fewer major opinion clusters for a compromise parameter of $\mu = 0.1$ than for $\mu = 0.3$ and $\mu = 0.5$. In the baseline DW model, we do not observe a noticeable difference in the number of major opinion clusters for different values of $\mu$.
For $(\mu, c) = (0.1,0.1)$, the number of major opinion clusters decreases as we decrease $\delta$ for both the $G(1000, 0.1)$ ER graphs and the $G(1000, 0.01)$ ER graphs. 
Additionally, for the $G(1000, 0.1)$ ER graphs and $(\mu, c) = (0.1,0.1)$, the number of major opinion clusters decreases as we decrease $z_0$.

{In our simulations on $G(1000,0.1)$ and $G(1000,0.01)$ ER graphs, the Shannon entropy follows the same trend as the number of major opinion clusters for both our adaptive edge-weighted DW model and the baseline DW model. 
However, as with the numbers of major opinion clusters, for small confidence bounds $c$ 
(specifically, for $c = 0.1$ or $c \in \{0.1, 0.2\}$), the Shannon entropy is (1) different for different values of $\mu$ in our adaptive edge-weighted DW model but is (2) similar for all values of $\mu$ in the baseline DW model.}

\begin{figure*}[htbp]
  \includegraphics[width=\linewidth]{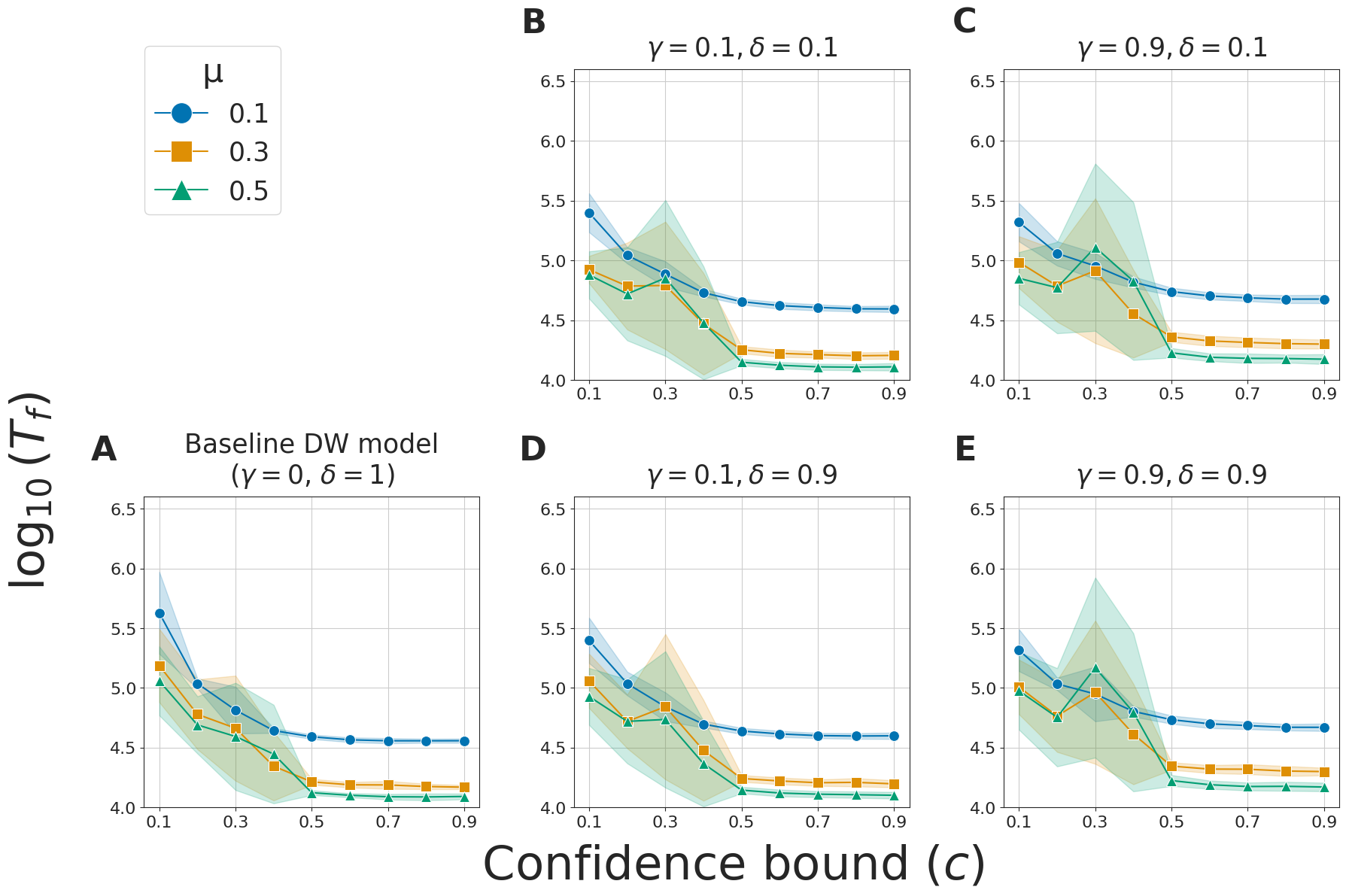}
  \caption{The logarithm of the convergence times of simulations of our adaptive edge-weighted DW model versus the confidence bound $c$ for different values of the edge-weight increase parameter $\gamma$ and the edge-weight decrease parameter 
  $\delta$ on $G(1000, 0.1)$ ER graphs. The initial edge-weight parameter is $z_0 = 0.1$.
  }
  \label{figure:ER_time_1}
\end{figure*}

The adaptivity of the edge weights in our model has
different effects on the convergence times of our simulations on the $G(1000, 0.1)$ ER graphs than on those for
the $G(1000, 0.01)$ ER graphs. 
Our simulations of our adaptive edge-weighted DW model on the $G(1000, 0.1)$ ER graphs yield similar trends as in our simulations on complete graphs with $N = 100$ and $N = 200$ nodes (see Figure \ref{figure:ER_time_1}). 
For a confidence bound of $c = 0.1$, our adaptive edge-weighted DW model takes less time to converge than the baseline DW model for all values of $\delta$ and $\gamma$ for the $G(1000, 0.1)$ E} graphs. 
When $z_0 = 0.1$, the convergence time has a larger variance in our adaptive edge-weighted DW model
than in the baseline DW model for $\mu \in \{0.3,0.5\}$ and $c \leq 0.5$.

\begin{figure*}[htbp]
  \includegraphics[width=\linewidth]{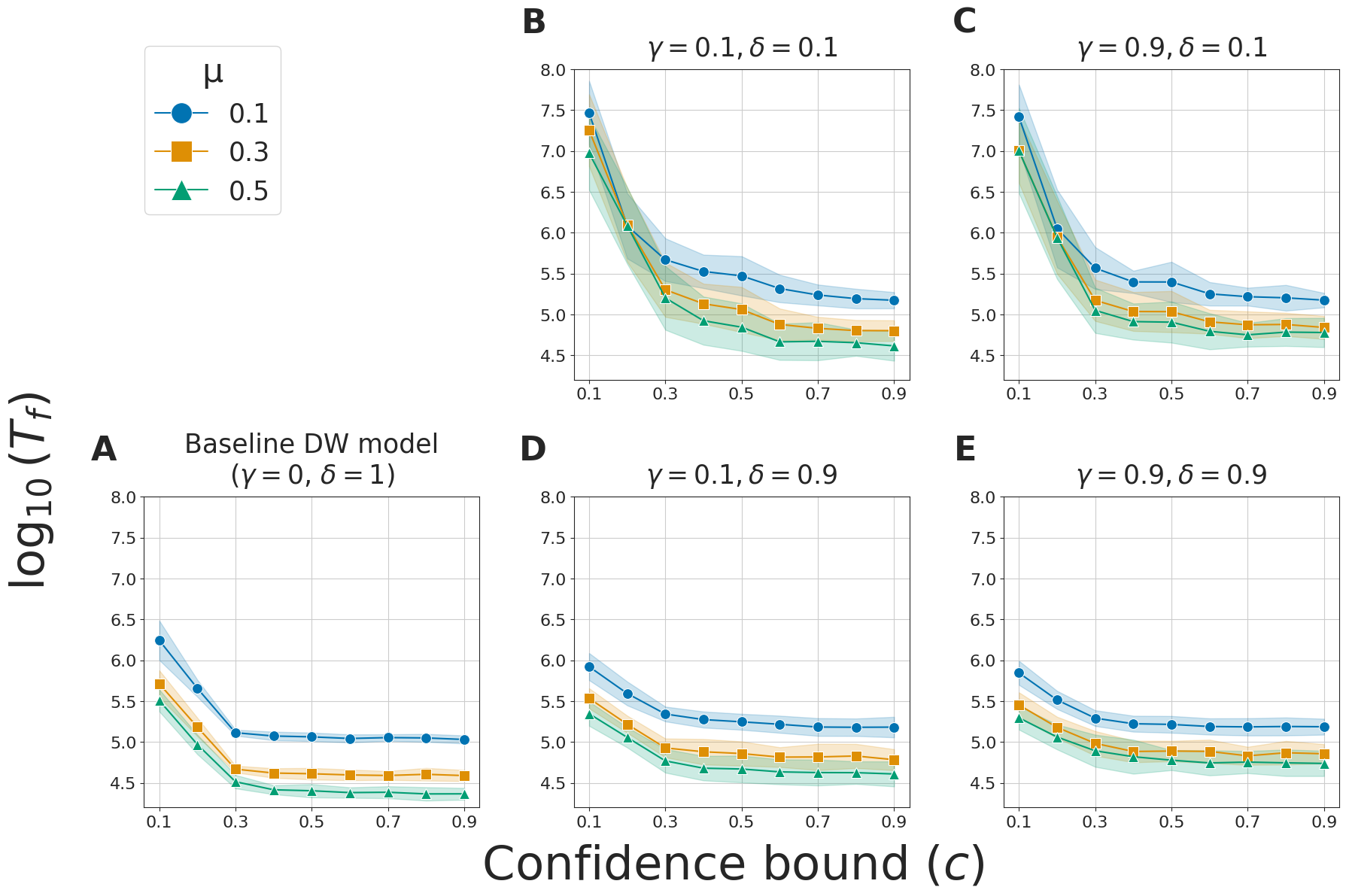}
  \caption{The logarithm of the convergence times of simulations of our adaptive edge-weighted DW model versus the confidence bound $c$ for different values of the edge-weight increase parameter $\gamma$ and the edge-weight decrease parameter 
  $\delta$ on $G(1000, 0.01)$ ER graphs. The initial edge-weight parameter is $z_0 = 0.1$.
   }
  \label{figure:ER_time_2}
\end{figure*}

{For $G(1000, 0.01)$ ER graphs}, we see {in Figure \ref{figure:ER_time_2}} that our adaptive edge-weighted DW model has a larger convergence-time variance than the baseline DW model for all values of the confidence bound $c$ 
{when $\delta = 0.1$}. 
When $\delta = 0.1$ and $c \leq 0.5$, we also observe that our adaptive edge-weighted DW model has longer convergence times than the baseline DW model. 
For small confidence bounds $c$ (specifically, {$c = 0.1$ and $c = 0.2$}), the convergence time decreases as we increase $\delta$. When $\delta = 0.9$ and $c \leq 0.2$, our adaptive edge-weighted DW model has a shorter convergence time than the baseline DW model. Additionally, for an initial edge-weight parameter of $z_0 = 0.1$, our adaptive edge-weighted DW model takes slightly longer to converge than the baseline DW model for $c \geq 0.3$ and all values of $\delta$.

Our simulations on the $G(100, 0.1)$ ER graphs behave similarly to our simulations on $G(1000, 0.01)$ ER graphs.


\subsection{{Degree-regular cycle-like graphs}}
\label{sec:deg-regular}

{To investigate the effect of edge density on the behavior of our adaptive edge-weighted DW model, we simulate the model on $k$-regular cycle-like graphs\map{.}
For $k$-regular cycle-like graphs of the same size, different degrees $k$ correspond directly to different edge densities. We consider degrees $k \in \{10,50,100,200\}$ for graphs with $N = 1000$ nodes. For all three of our summary statistics 
(namely, the number of major opinion clusters, the Shannon entropy, and the convergence time), our simulations have different trends on the $10$-regular cycle-like graph than on $k$-regular cycle-like graphs with $k \in \{50,100,200\}$.}

For small confidence bounds $c$ (specifically, $c \leq 0.2$), the number of major opinion clusters increases as we decrease $k$ in both our adaptive edge-weighted DW model and the baseline DW model. In particular, we observe significantly more major opinion clusters for $k = 10$ than for all other values of $k$. Additionally, our simulations on {the $10$-regular cycle-like graph} have peaks in the number of major opinion clusters for $c = 0.2$. That is, the number increases from $c = 0.1$ to $c = 0.2$ and then decreases from $c = 0.2$ to $c = 0.3$. We do not observe such peaks for the other $k$-regular cycle-like graphs. 
For $c \geq 0.3$, our simulations on all of the examined $k$-regular {cycle-like} graphs tend to achieve consensus.

\begin{figure*}[htbp]
  \includegraphics[width=\linewidth]{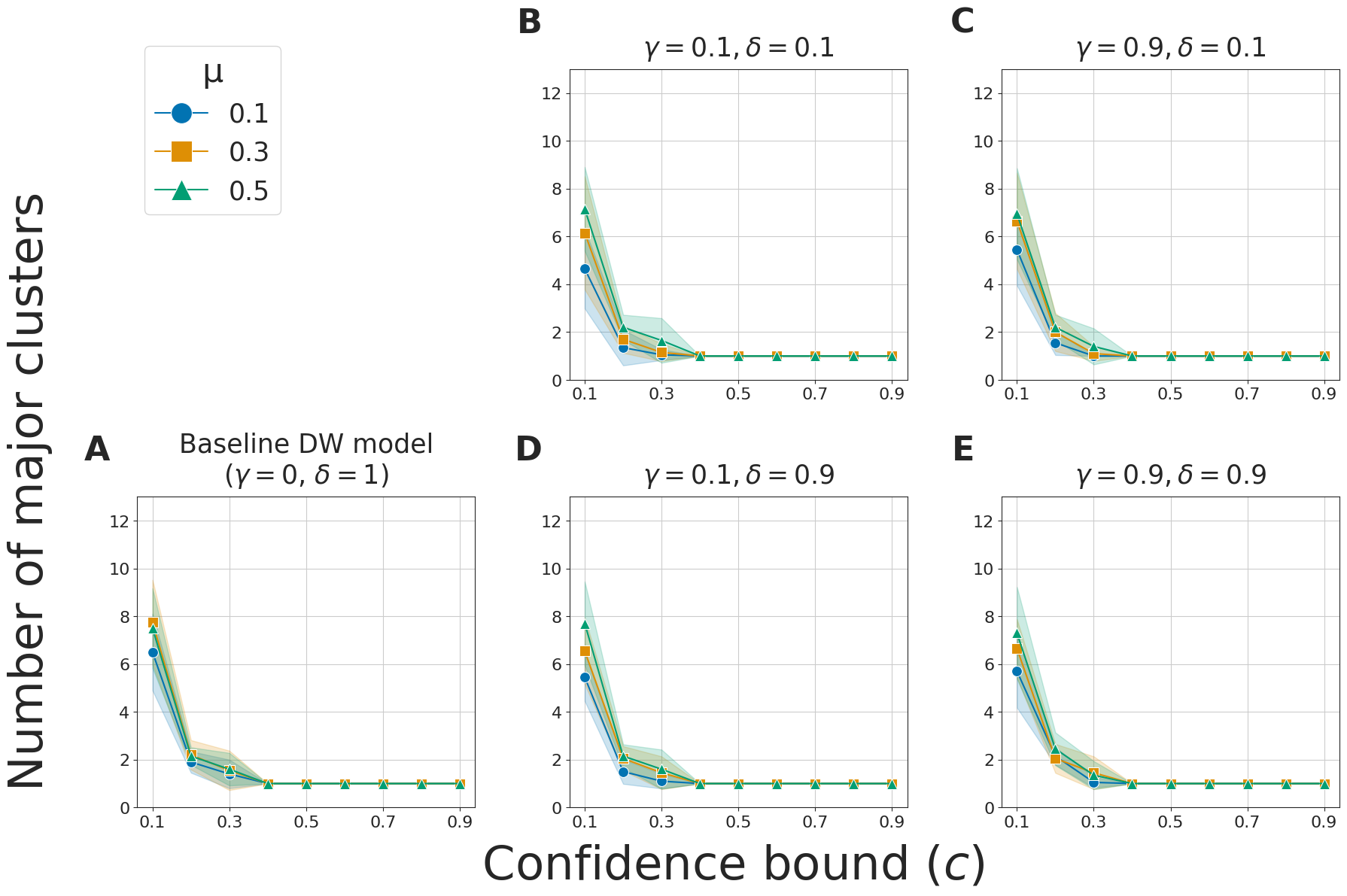}
  \caption{The number of major opinion clusters in simulations of our adaptive edge-weighted DW model versus the confidence bound $c$ for different values of the edge-weight increase parameter $\gamma$ and the edge-weight decrease parameter $\delta$ on a 50-regular {cycle-like} graph.
The initial edge-weight parameter is $z_0 = 0.1$.
  }
  \label{figure:deg_opinions_1}
\end{figure*}

\begin{figure*}[htbp]
  \includegraphics[width=\linewidth]{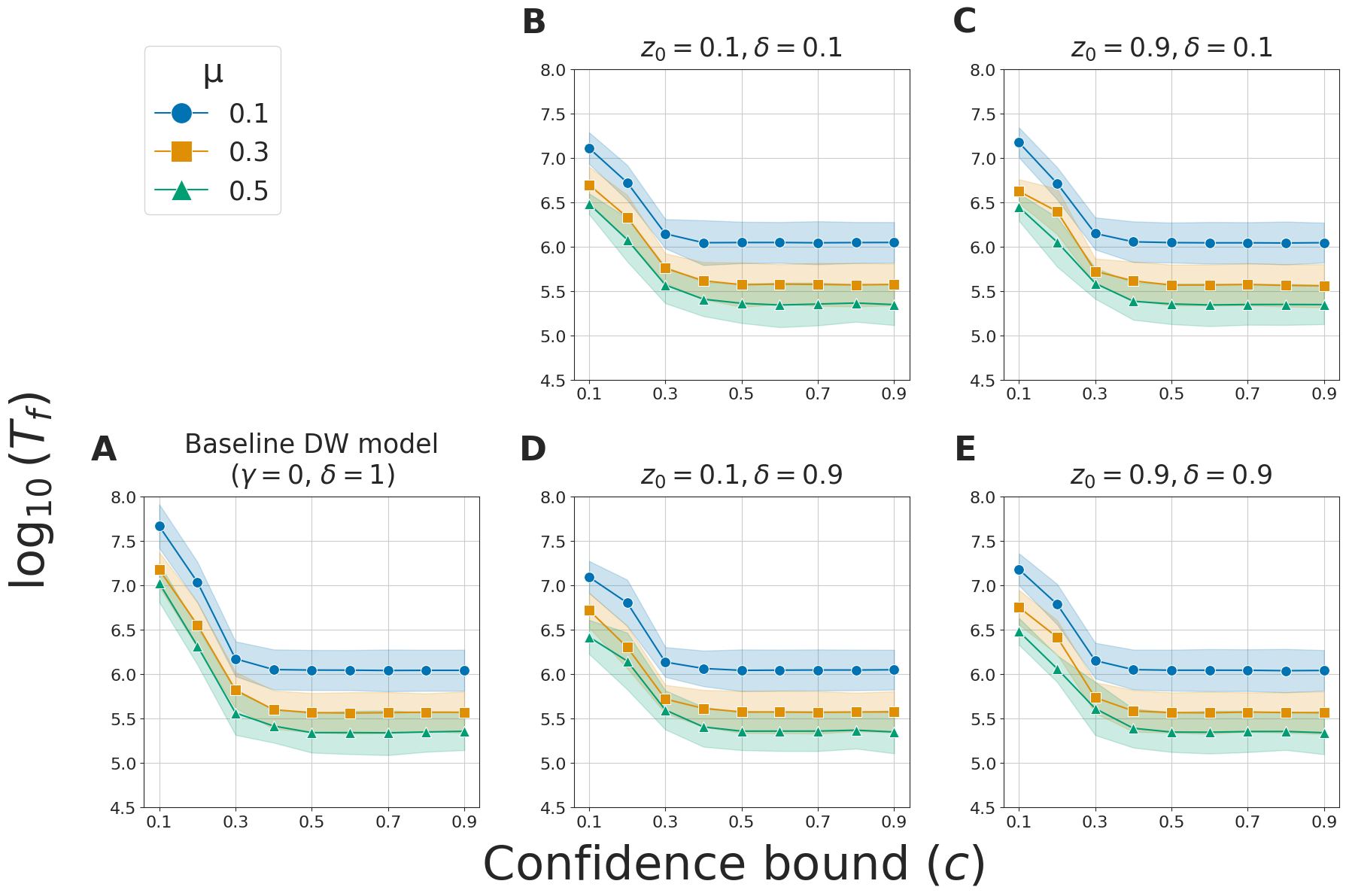}
  \caption{The logarithm of the convergence times in simulations of our adaptive edge-weighted DW model versus the confidence bound $c$ for different values of the edge-weight increase parameter $\gamma$ and the edge-weight decrease parameter $\delta$ on a 50-regular cycle-like graph. The initial edge-weight parameter is $z_0 = 0.1$.
  }
  \label{figure:deg_time_1}
\end{figure*}

\begin{figure*}[htbp]
  \includegraphics[width=\linewidth]{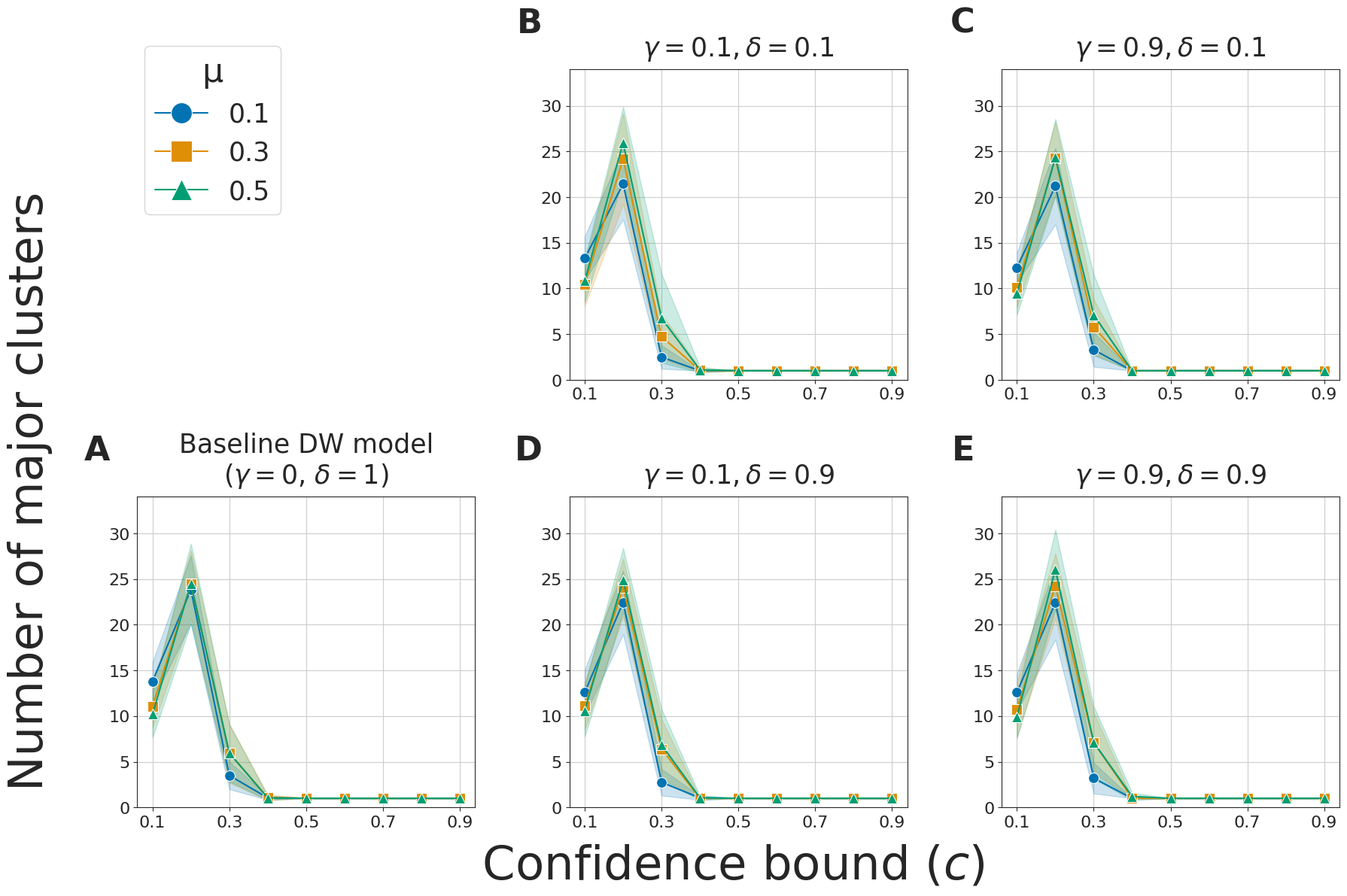}
  \caption{The numbers of major opinion clusters in simulations of our adaptive edge-weighted DW model versus the confidence bound $c$ for different values of the edge-weight increase parameter $\gamma$ and the edge-weight decrease parameter $\delta$ on a 10-regular {cycle-like} graph. The initial edge-weight parameter is $z_0 = 0.1$.
  }
  \label{figure:deg_opinions_2}
\end{figure*}

When $k \in \{50,100,200\}$ and $c = 0.1$, we obtain fewer major opinion clusters in our adaptive edge-weighted DW model for $\mu = 0.1$ than for $\mu \in \{0.3, 0.5\}$ (see Figure \ref{figure:deg_opinions_1}). 
 For $(\mu, c) = (0.1,0.1)$, the number of major opinion clusters tends to decrease as we decrease $\delta$. 
 Although our simulations of the baseline DW model sometimes also have fewer major opinion clusters for $\mu = 0.1$ than for {$\mu \in \{0.3, 0.5\}$}, the drop-off as we decrease $\mu$ is sharper in our adaptive edge-weighted DW model than in the baseline DW model. When $k = 10$ and $c = 0.2$, there are slightly fewer major opinion clusters for $\mu = 0.1$ than for $\mu \in \{0.3, 0.5\}$ {in our adaptive edge-weighted DW model}.
 By contrast, {when $k = 10$ and $c = 0.2$}, we do not observe a noticeable difference in the numbers of major opinion clusters for different values of $\mu$ in the baseline DW model (see Figure \ref{figure:deg_opinions_2}).

\begin{figure*}[htbp]
  \includegraphics[width=\linewidth]{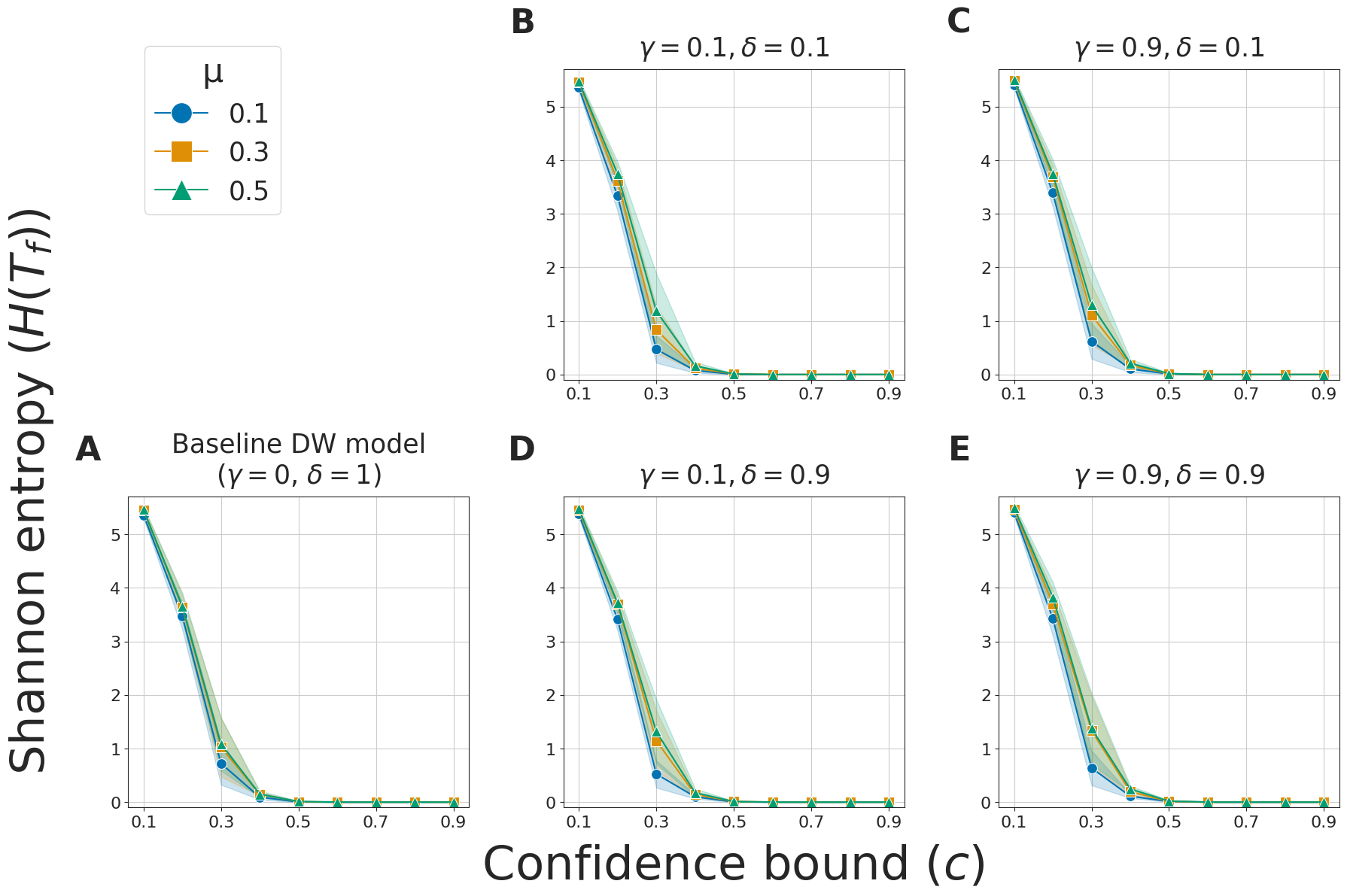}
  \caption{The Shannon entropies {at the convergence time $T_f$} in simulations of our adaptive edge-weighted DW model versus the confidence bound $c$ for different values of the edge-weight increase parameter $\gamma$ and the edge-weight decrease parameter $\delta$ on a 10-regular cycle-like graph. The initial edge-weight parameter is $z_0 = 0.1$.
  }
  \label{figure:deg_entropy}
\end{figure*}

As with the number of major opinion clusters, for small confidence bounds $c$ (specifically, {$c = 0.1$ and $c = 0.2$}), the Shannon entropies in simulations of our adaptive edge-weighted DW model and the baseline DW model both increase as we decrease $k$. Additionally, the Shannon entropy is significantly larger for $k = 10$ than for the other values of $k$. 

{For degrees $k \in \{50,100,200\}$, the Shannon entropy and the number of major opinion clusters have similar trends. 
{However}, when $k = 10$, the Shannon entropy decreases from $c = 0.1$ to $c = 0.2$ (see Figure \ref{figure:deg_entropy}}) and the number of major opinion clusters increases. 
This arises because there are many more minor opinion clusters for $c = 0.1$ than for $c = 0.2$.
The nodes in the $10$-regular cycle-like graph have fewer neighbors than when $k \in \{50,100,200\}$, so each node is receptive to a smaller maximum number of nodes.
Therefore, when there is opinion fragmentation, we expect the final opinion clusters in the $10$-regular cycle-like graph to be smaller than those in $k$-regular cycle graphs with $k\in \{50,100,200\}$.
The small sizes of the opinion clusters results in a large number of minor opinion clusters when $c = 0.1$ and $k = 10$.

\begin{figure*}[htbp]
  \includegraphics[width=\linewidth]{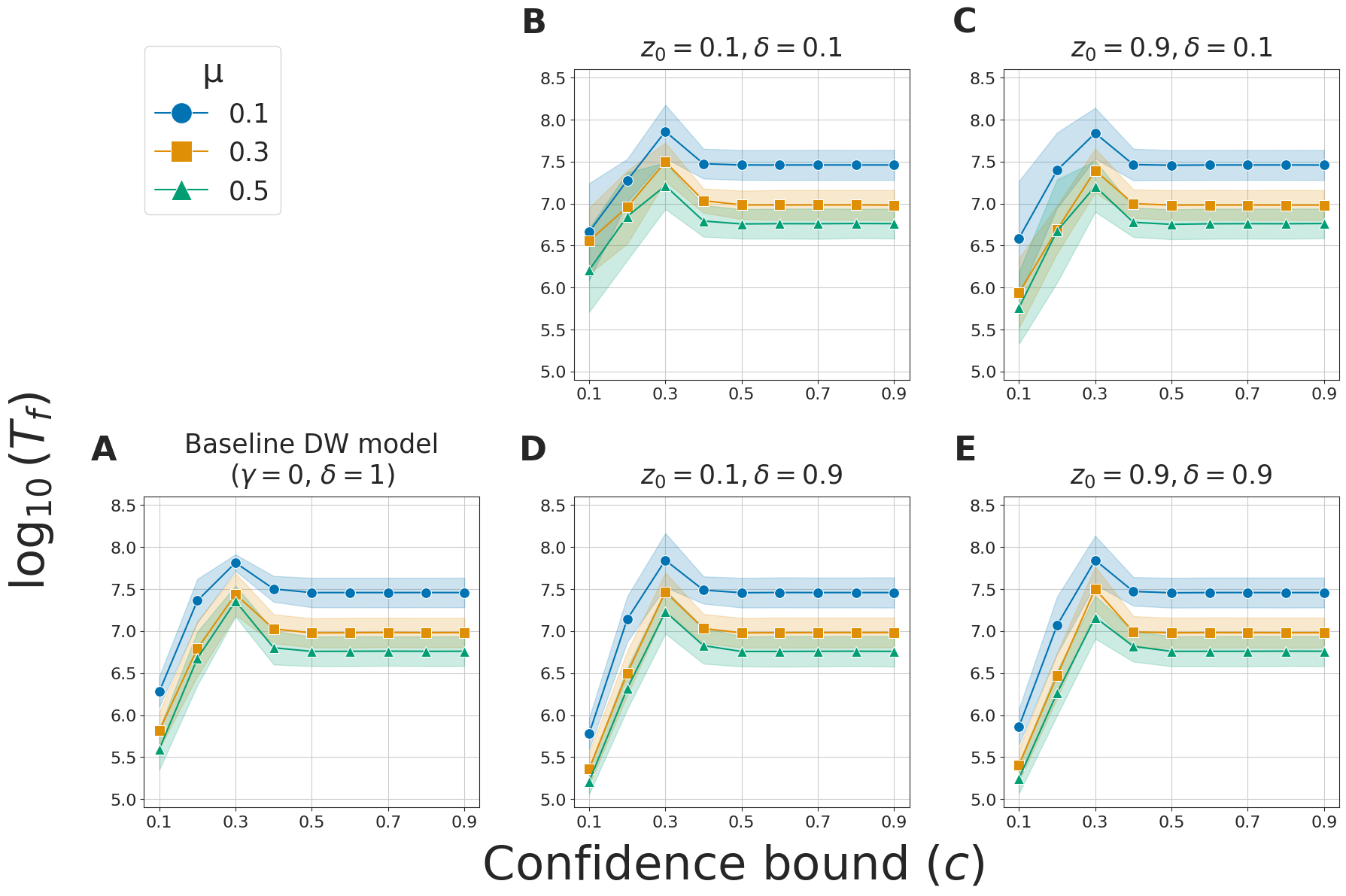}
  \caption{The logarithm of the convergence times in simulations of our adaptive edge-weighted DW model versus the confidence bound $c$ for different values of the initial edge-weight parameter $z_0$ and the edge-weight decrease parameter $\delta$ on a 10-regular cycle-like graph. The edge-weight increase parameter is $\gamma = 0.1$.
  }
  \label{figure:deg_time_2}
\end{figure*}

The convergence times of our simulations of our adaptive edge-weighted DW model on a 10-regular cycle-like graph peak for a different confidence bound $c$ than our simulations on $k$-regular cycle-like graphs with $k \in \{50,100,200\}$.
{The peak is at $c = 0.2$ for the 10-regular {cycle-like} graph, but the peak is at $c = 0.1$ for $k \in \{50,100,200\}$.}
Additionally, for $k \in \{50,100,200\}$, {the trends in the convergence-time properties} of our adaptive edge-weighted DW model tend to align with the convergence-time trends
that we observed in our simulations on complete graphs.
{In particular,} as in complete graphs, our adaptive edge-weighted DW model has smaller convergence times than the baseline DW model for $c \leq 0.2$ {when $k \in \{50,100,200\}$} (see Figure \ref{figure:deg_time_1}).
{{By contrast, the {trends in the convergence-time properties of} our adaptive edge-weighted DW model on a $10$-regular cycle-like graph are} similar to the convergence-time trends for the $G(1000, 0.01)$ ER graphs.}
When $k = 10$ and $\delta = 0.1$, we observe both that the convergence-time variance
is larger in our adaptive edge-weighted DW model than in the baseline DW model for all values of $c$ and that our adaptive edge-weighted DW model takes longer to converge than the baseline DW model when $c = 0.1$. 
For $k = 10$ and $\delta > 0.1$, our adaptive edge-weighted DW model converges faster than the baseline DW model for $c \leq 0.2$ (see Figure \ref{figure:deg_time_2}).

{As we have just discussed, the incorporation of adaptive edge weights has different effects on the convergence-time properties of our simulations on the $10$-regular cycle-like graph than on those of our simulations on $k$-regular cycle-like graphs with $k \in \{50,100,200\}$. We hypothesize that this difference 
arises from the graphs' different mean degrees and hence from their different edge densities.
More generally, for all of the graph types that we consider,
our adaptive edge-weighted DW model consistently has different convergence-time properties for dense graphs than for sparse graphs.}


\subsection{Network of network scientists ({\sc NetScience})} 
\label{sec:netscience}

We also simulate our adaptive edge-weighted DW model on the \textsc{NetScience} network~\cite{Netscience}, which is a coauthorship network between researchers in network science. In this unweighted and undirected graph, the nodes represent network scientists and the edges represent coauthorships between them.

As we increase the confidence bound {from $c = 0.2$ to $c = 0.5$}, we observe similar trends in the number of major opinion clusters and the Shannon entropy in our adaptive edge-weighted DW model for all values of the compromise parameter $\mu$, the initial edge-weight parameter $z_0$, the edge-weight increase parameter $\gamma$, and the edge-weight decrease parameter $\delta$. These trends are broadly similar to those in the baseline DW model.
{Both} our adaptive edge-weighted DW model and the baseline DW model experience sharp declines in both the number of major opinion clusters and the Shannon entropy as we increase the confidence bound {from $c = 0.2$ to $c = 0.5$}. 
{For $c \geq 0.5$, both models reach consensus.}

As we increase the confidence bound {from $c = 0.1$ to $c = 0.2$}, the number of major opinion clusters and the Shannon entropy have different trends from each other in both our adaptive edge-weighted DW model and the baseline DW model for all values of the parameters $\mu$, $z_0$, $\gamma$, and $\delta$.
The number of major opinion clusters increases initially and then decreases as we increase $c$, whereas the Shannon entropy decreases as we increase $c$.
However, this increase in the number of major opinion clusters does not entail a decrease in fragmentation. 
As we discussed in Section \ref{sec:deg-regular}, there is a similar discrepancy between the trends in the number of major opinion clusters and the Shannon entropy in the $10$-regular cycle-like graph.
As in the $10$-regular cycle-like graph, we believe that the discrepancy between the trends in the Shannon entropy and the number of major opinion clusters arises from the presence of many small opinion clusters. 
The \textsc{NetScience} network has many small-degree nodes and a similar edge density to that of the $10$-regular cycle-like graph. 
The low edge density of the \textsc{NetScience} network limits the sizes of the nodes' receptive neighborhoods and leads to smaller opinion cluster sizes when fragmentation occurs. 
Consequently, there are many minor opinion clusters when $c = 0.1$.

Our simulations on the \textsc{NetScience} network also have other notable features. For example, we observe more opinion fragmentation for $(c, z_0) = (0.2, 0.1)$ in our adaptive edge-weighted DW model than in the baseline DW model for all values of $\mu$, $\gamma$, and $\delta$.
For $c = 0.3$, either $\gamma = 0.5$ or $\delta = 0.5$, and all values of the initial edge-weight parameter $z_0$, we observe that the number of major opinion clusters is smaller for $\mu = 0.1$ than for $\mu \in\{0.3,0.5\}$ in our adaptive edge-weighted DW model.
By contrast, for these values of {$(\gamma, \delta, c, z_0)$,} we do not observe a noticeable difference in the number of major opinion clusters for different values of $\mu$ in the baseline DW model.

The convergence time increases as a function of the confidence bound $c$ for $c \leq 0.3$, and the convergence is slowest near $c = 0.3$ for all values of the initial edge-weight parameter $z_0$ for both our adaptive edge-weighted DW model and the baseline DW model. 
We believe that the initial increase in convergence times with the confidence bound in the \textsc{NetScience} network (as we also observed for the 10-regular cycle-like graph in Section \ref{sec:deg-regular}) likely arises both from its sparsity and from the presence of many small-degree nodes.

For $\delta = 0.1$ and $c \leq 0.6$, our adaptive edge-weighted DW model has a larger convergence-time variance than the baseline DW model for all values of the parameters $\mu$, $z_0$, and $\gamma$.
We hypothesize that this trend in the convergence-time variance arises from the \textsc{NetScience} network's sparsity and its abundance of small-degree nodes.
We {observed} similar convergence-time trends in the $G(1000, 0.01)$ ER graphs (see Section \ref{sec:ER}) and the $10$-regular cycle-like graph (see Section \ref{sec:deg-regular}), which have similar edge densities. The $G(1000, 0.01)$ ER graphs have an expected edge density of $0.01$, and the $10$-regular cycle-like graph has an edge density of $10/999 \approx 0.01$.
See Section \ref{sec:deg-regular} for a more detailed explanation of our hypothesis.


\subsection{Facebook friendship networks} \label{sec:facebook}

We also simulate our adaptive edge-weighted DW model on the largest connected components of four {graphs} from the \textsc{Facebook100} data set \cite{Facebook100} (see Table \ref{table:Table1}). These graphs are single-university networks of Facebook friendships between individuals at {California Institute of Technology (Caltech)}, Reed College, Swarthmore College, and Smith College on one day in fall 2005.
{The Caltech network has 762 nodes, the Reed College network has 962 nodes, the Swarthmore College network has 1657 nodes, and the Smith College network has 2970 nodes.}

\begin{figure*}[htbp]
  \includegraphics[width=\linewidth]{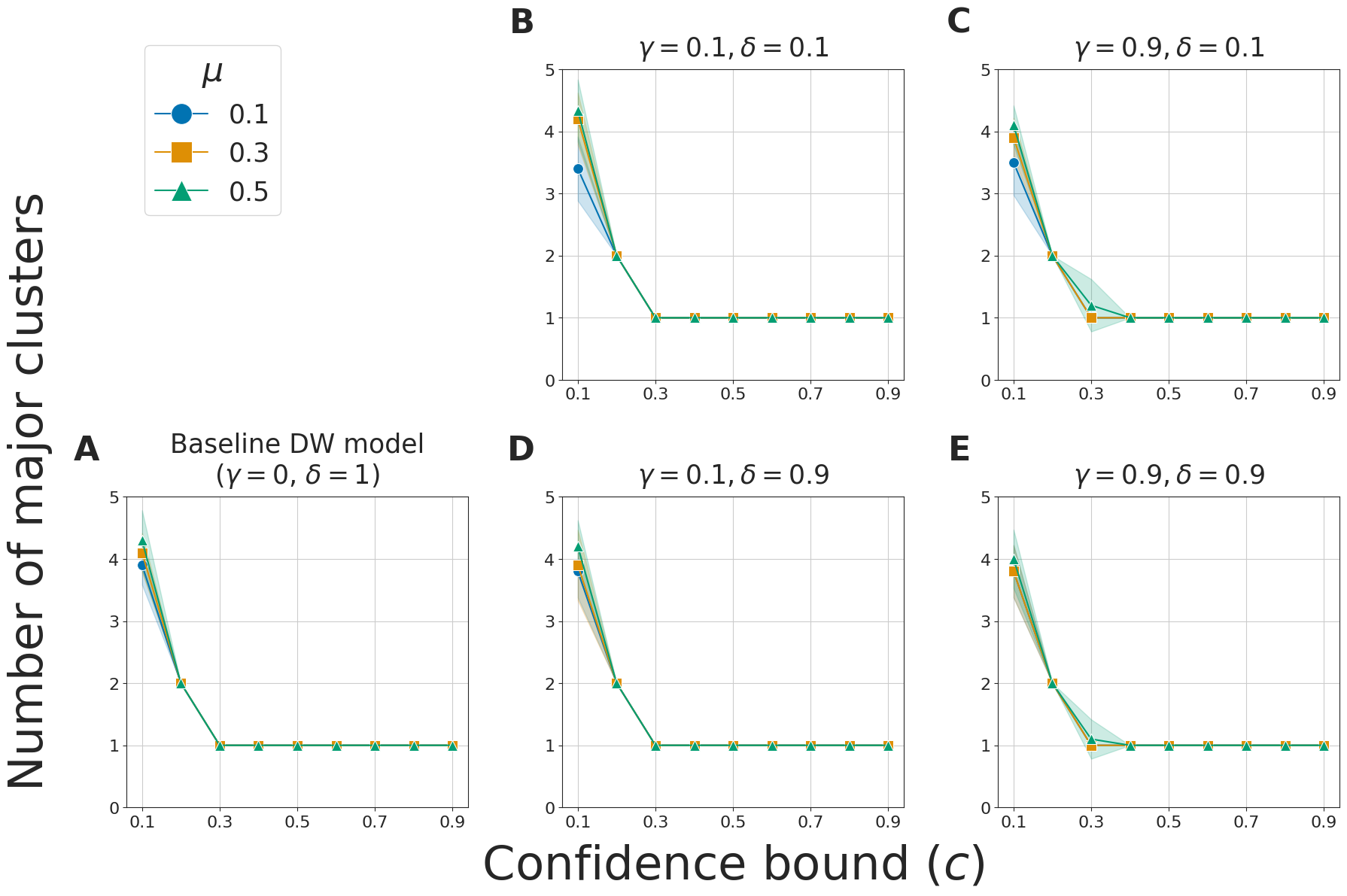}
  \caption{The number of major opinion clusters in simulations of our adaptive edge-weighted DW model versus the confidence bound $c$ for different values of the edge-weight increase parameter $\gamma$ and the edge-weight decrease parameter $\delta$ on the Smith College network. The initial edge-weight parameter is $z_0 = 0.1$.
  }
  \label{figure:smith_opinions}
\end{figure*}

For all four \textsc{Facebook100} networks, our adaptive edge-weighted DW model has similar general trends to those of the complete graphs in both the number of major opinion clusters and the Shannon entropy.
For confidence bounds $c \in \{0.1,0.2,0.3,0.4,0.5\}$ and all values of the initial edge-weight parameter $z_0$, the edge-weight increase parameter $\gamma$, and the edge-weight decrease parameter $\delta$, the Shannon entropy is noticeably smaller when the compromise parameter is $\mu = 0.1$ than when $\mu = 0.3$ or $\mu = 0.5$.
By contrast, the number of major opinion clusters is smaller for $\mu = 0.1$ than for $\mu = 0.3$ or $\mu = 0.5$ for a smaller set of confidence bounds $c$. 
This set of confidence bounds is different for different \textsc{Facebook100} networks and for different values of $\gamma$ and $\delta$.
Additionally, for $c \in \{0.1, 0.3\}$, the number of major opinion clusters has the largest variance for the Caltech network, the Reed College network, and the Swarthmore College network for both our adaptive edge-weighted DW model and the baseline DW model.
The confidence bound $c$ that yields the largest variance in the number of major opinion clusters is not always the same for the Smith College network.
For the baseline DW model, we observe the largest variance in the number of major opinion clusters when $c = 0.1$ for all values of {$\mu$, $z_0$, $\gamma$, and $\delta$}.
For our adaptive edge-weighted DW model, the largest variance sometimes occurs when $c = 0.1$, sometimes occurs when $c = 0.3$, and sometimes occurs when $c \in \{0.1, 0.3\}$.
{We show the number of major opinion clusters for the Smith College network in Figure \ref{figure:smith_opinions}. Analogous figures for the other \textsc{Facebook100} networks are available in our repository \cite{github}.}

As we increase the confidence bound $c$, the convergence time of the baseline DW model decreases slightly for $c \leq 0.3$ and changes little for $c \geq 0.3$. 
By contrast, in our adaptive edge-weighted DW model, when {the edge-weight decrease parameter} $\delta > 0.1$, the convergence time tends to increase as we increase $c$ {for $c \leq 0.3$}.
In our simulations on the \textsc{Facebook100} networks, the differences in the convergence-time trends between our adaptive edge-weighted DW model and the baseline DW model are similar to those in our simulations on the $G(1000, 0.01)$ ER graphs, the $10$-regular cycle-like graph, and the \textsc{NetScience} network.
We hypothesize that these differences in trends between our adaptive edge-weighted DW model and the baseline DW model are influenced by the sparse nature of the \textsc{Facebook100} networks. All of the \textsc{Facebook100} networks have edge densities between $0.02$--$0.06$.


\section{Conclusions and discussion} \label{sec:conclusions}

We formulated and studied a bounded-confidence model (BCM) of opinion dynamics in which the interaction frequencies between pairs of agents depend on their previous interactions.
Our model, which is a generalization of the Deffuant--Weisbuch (DW) BCM, has adaptive edge weights that govern the interaction probabilities between agents. 
One can view these edge weights as encoding the strengths of the ties between agents. 
By incorporating adaptivity into edge weights, we investigated how prior interactions between agents help shape how often pairs of agents interact with each other in the future. 

We proved results for the long-time dynamics of the edge weights and the limiting behavior of effective graphs, which consist of the edges of a graph that encode social ties between mutually receptive agents, in our adaptive edge-weighted DW model. 
We showed that the edge weights almost surely converge to either 0 or 1, which respectively entail the smallest and largest interaction probabilities between agents.
Moreover, if each edge weight converges either to 0 or to 1, we proved that the edge weight of an edge between two agents almost surely converges to 1 if they have the same limit opinion and almost surely converges to 0 if they have different limit opinions (see Theorem \ref{theorem:edge_convergence}). 
 When the limit opinions of {adjacent} agents never differ by precisely the confidence bound, we also proved that the effective graphs in our adaptive edge-weighted DW model are guaranteed to converge (see Theorem \ref{thm:effgraph}).

Guided by our analysis of effective graphs, we also conducted numerical simulations of our adaptive edge-weighted DW model on several types of synthetic graphs and {real-world networks.} 
In our simulations, we examined the numbers of major opinion clusters, the Shannon entropies, and the convergence times for each of these {graphs}. 
We observed that the differences, especially in the convergence times, in qualitative dynamics between our adaptive edge-weighted DW model and the baseline DW model depend on the type of {graph}.

It is worthwhile to systematically investigate why incorporating adaptive edge weights affects the DW dynamics differently for different types of graphs. We also conducted {a small set of} simulations (which we did not present) of our adaptive edge-weighted DW model on two-community stochastic-block-model (SBM) graphs~\cite{newman2018}, and we obtained similar behavior as in our simulations on complete graphs. 
Therefore, the observed differences in dynamics do not appear to be due to community structure. Instead, given the observed differences in our adaptive edge-weighted DW model's dynamics between 10-regular cycle-like graphs and $k$-regular cycle-like graphs with larger $k$, we hypothesize that these differences may be due in part to the difference in the mean degrees of the graphs. However, it is important to conduct further simulations on a variety of {graphs} to examine this hypothesis.

There are also other interesting phenomena that one can explore in our adaptive edge-weighted DW model. Our simulations sometimes yielded different dynamics for different values of the initial edge weights, which we assumed are the same for all agents. It is worthwhile to conduct simulations with a larger set of initial edge weights to better understand how they affect our model's qualitative dynamics. 
Moreover, in line with recent research on node-weighted BCMs \cite{grace2022}, it is worthwhile to consider the effects of {heterogeneous initial edge weights (and hence heterogeneous initial interaction probabilities).}
Additionally, one can examine adaptive node weights, which one can interpret as changes in agents' sociabilities or activities level after their interactions. It is worthwhile to compare the behavior of our adaptive edge-weighted DW model with adaptive node-weighted DW models.

Opinion models with adaptive edge weights may be useful in studies of the effects of social-media algorithms on opinion dynamics and social influence. Researchers have used BCMs to examine algorithmic bias \cite{alg_bias} and the impact of media accounts on opinion dynamics on social-media platforms \cite{media_nodes}. Some social-media algorithms use prior interactions of users with content to increase the probability that they see similar material in the future \cite{tiktok, instagram}. 
However, unlike in our adaptive edge-weighted DW model, many of these algorithms do not distinguish between ``positive" interactions (i.e., interactions with content to which a user is receptive) and ``negative" interactions (i.e., interactions with content to which a user is not receptive) when they increase the probability that a user sees similar content in the future {(see, e.g., \cite{milli2025})}. 
Future research efforts can explore the effect of algorithmic content recommendations on opinion dynamics on social-media platforms by using adaptive edge weights with various {opinion update rules}.
\vspace{2mm}



\begin{acknowledgements}
\vspace{-2mm}

{JL acknowledges funding from the National Science Foundation (grant numbers 1829071 and 2136090). MAP and JL acknowledge funding from the National Science Foundation (grant number 1922952) through the Algorithms for Threat Detection (ATD) program.}

\end{acknowledgements}


\vspace{-3mm}
\section*{Data availability}
\vspace{-2mm}

The data from our simulations, which is the only data in the present paper, is not available. The code 
to generate all of our simulations is available at Ref.~\cite{github}. This repository also includes plots from additional numerical simulations.




%


\end{document}